\documentclass[10pt,twocolumn,notitlepage,aps,pra,showpacs,superscriptaddress,longbibliography,nofootinbib]{revtex4-2} 

\usepackage[T1]{fontenc}% optional T1 font encoding

\usepackage[utf8]{inputenc} 
\usepackage{amssymb,mathtools}
\usepackage{graphicx}
\usepackage{amsthm}
\usepackage{amssymb, latexsym, bm, mathtools, braket, multirow,enumitem}

\usepackage[cmintegrals]{newtxmath}
\usepackage{color}

\DeclareFontFamily{OMX}{MnSymbolE}{}
\DeclareFontShape{OMX}{MnSymbolE}{m}{n}{
    <-6>  MnSymbolE5
   <6-7>  MnSymbolE6
   <7-8>  MnSymbolE7
   <8-9>  MnSymbolE8
   <9-10> MnSymbolE9
  <10-12> MnSymbolE10
  <12->   MnSymbolE12}{}
\DeclareSymbolFont{mnlargesymbols}{OMX}{MnSymbolE}{m}{n}
\SetSymbolFont{mnlargesymbols}{bold}{OMX}{MnSymbolE}{b}{n}
\DeclareMathDelimiter{\llangle}{\mathopen}{mnlargesymbols}{'164}{mnlargesymbols}{'164}
\DeclareMathDelimiter{\rrangle}{\mathclose}{mnlargesymbols}{'171}{mnlargesymbols}{'171}

\theoremstyle{plain}
\newtheorem{lemm}{Lemma}

\newtheorem{theo}{Theorem}

\theoremstyle{definition}

\newtheorem{prot}{Protocol}

\newcommand{\cB}{\mathcal{B}}

\newcommand*{\QEDA}{\hfill\ensuremath{\blacksquare}}%
\newcommand{\ii}{\iota}

\newcommand{\pureset}[1]{\mathrm{P}(#1)}

\newcommand{\Rr}{\sR}
\newcommand{\Ss}{\sS}

\newcommand\sff{\mathsf{f}}

\newcommand\sF{\mathsf{F}}

\newcommand\sA{\mathsf{A}}
\newcommand\sB{\mathsf{B}}
\newcommand\sC{\mathsf{C}}

\newcommand\sI{\mathsf{I}}
\newcommand\sM{\mathsf{M}}
\newcommand\sR{\mathsf{R}}
\newcommand\sS{\mathsf{S}}
\newcommand\sU{\mathsf{U}}
\newcommand\sV{\mathsf{V}}
\newcommand\sX{\mathsf{X}}
\newcommand\sY{\mathsf{Y}}
\newcommand\sZ{\mathsf{Z}}
\newcommand\sT{\mathsf{T}}

\usepackage{tikz}
\usetikzlibrary{arrows}
\usetikzlibrary{shapes,snakes, quantikz}
\usetikzlibrary{tikzmark}
\usetikzlibrary{fit,backgrounds}
\usetikzlibrary{matrix,decorations.pathreplacing,angles,quotes}
\usetikzlibrary{patterns}
\usetikzlibrary{calc}
\usetikzlibrary{positioning}
\tikzstyle{block} = [rectangle, draw, 
    minimum width=3em, text centered, minimum height=12em]
\tikzstyle{rblock} = [rectangle, draw, 
    minimum width=5em, text centered, minimum height=4em, rounded corners]   
\tikzstyle{circ} = [circle, draw, inner sep=0em, minimum width=2em,
     text centered]
\tikzstyle{circ2} = [circle, draw, inner sep=0em, minimum width=2.2em,
     text centered]    
\tikzstyle{rec} = [rectangle, draw]
\tikzstyle{line} = [draw, -latex]
\tikzstyle{imp-line} = [draw, -implies, double equal sign distance]

\tikzset{snake arrow/.style=
{
decorate,
decoration={snake,amplitude=.4mm,segment length=2mm,pre length=1mm, post length=1mm}}
}
\tikzset{snake arrow0/.style=
{
decorate,
decoration={snake,amplitude=.4mm,segment length=2mm,post length=1mm}}
}

\def\Label#1{\label{#1}\ [\ \text{#1}\ ]\ }
\def\Label{\label}

\begin{document}
\title{Two-Server Oblivious Transfer for Quantum Messages}
\author{Masahito Hayashi}
\email{hayashi@sustech.edu.cn}
\affiliation{Shenzhen Institute for Quantum Science and Engineering, Southern University of Science and Technology, Shenzhen,518055, China}
\affiliation{International Quantum Academy (SIQA), Futian District, Shenzhen 518048, China}
\affiliation{Graduate School of Mathematics, Nagoya University, Nagoya, 464-8602, Japan}
\author{Seunghoan~Song}
\email{seunghoans@gmail.com}
\affiliation{Graduate School of Mathematics, Nagoya University, Nagoya, 464-8602, Japan}

%\thanks{This article was presented in part at Proceedings of 2021 IEEE International Symposium on Information Theory \cite{SH21}. }% <-this % stops a 

\begin{abstract}
Oblivious transfer is considered as a cryptographic primitive task for quantum information processing over quantum network.
Although it is possible with two servers, any existing protocol works only with classical messages.
We propose two-server oblivious transfer protocols for quantum messages.
\end{abstract}

\maketitle

\if0
\fi
\section{Introduction}

Construction of quantum internet is the ultimate goal of quantum technology \cite{Kimble}.
Quantum state communication and entanglement sharing over long distances are basic functions of the quantum internet, but to fully extract its potential, it is essential to find out its applications.
One promising direction is distributed quantum protocols and algorithms with security.
For example, quantum network coding has been studied from theory to experiment \cite{Kimble,Hayashi2007,PhysRevA.76.040301,Kobayashi2009,Leung2010,JFM11,SH18-2,PhysRevLett.101.060401,PhysRevA.80.022339, HS20,BH20,LIY19,PCXLY21,NBA17,PMS20,PXP21,PCX21,WE21} and blind quantum computation has also been extensively studied \cite{Childs,BFK,BKBF,MF,Morimae,MDF,MF2,LCWW,SZ,HM}.

As another distributed secure quantum protocol, this paper studies two-server quantum oblivious transfer (TQOT) for the transmission of quantum states.
Oblivious transfer
	is the task that the user downloads the intended message
	among several messages from the servers under two requirements.
As the first condition, the user's choice of the intended message is not leaked to the servers, which is called the user secrecy,
when the user is honest and the servers make arbitrary operations.
As the second condition,
the information of other messages is not leaked to the user, which is called
the server secrecy,
when the servers are honest and the user makes arbitrary operations.
It is known that one-server oblivious transfer is impossible even with the quantum system if we have no assumption \cite{Lo}.
However, if there are two servers that do not communicate with each other, oblivious transfer is possible when the messages are given as classical information.
That is, two-server oblivious transfer for classical messages (C-TOT) is available,
and is often called two-server symmetric private information retrieval (SPIR) for classical messages.

Although two-server oblivious transfer for classical messages is possible by using classical communication,
the use of quantum communication improves its communication speed.
In the following, this problem setting with quantum communication
is simplified to
classical two-server quantum oblivious transfer (C-TQOT).
Several studies were done on this problem when the message is classical information and
a noiseless quantum channel is available.
For example, Kerenidis and de Wolf \cite{KdW03, KdW04} studied this problem by relaxing the secrecy criterion.
When the number $\sff$ of messages is fixed,
the preceding study \cite{SH19} derived the optimal transmission rate for this problem, which is defined similarly to its classical counterpart
\cite{SJ17,SJ17-2} as the optimal communication efficiency for arbitrary-long classical messages.
It proved that a protocol
can be constructed
without any communication loss
when the message is classical information, a noiseless quantum channel is available, and
prior-entanglement among servers is allowed.
The papers \cite{SH19-2,SH20}
and Allaix et al. \cite{AHPH20, ASHPHH21}
also considered this problem
with colluding servers in which
secrecy of the protocol is preserved even if some servers may communicate and collude.
Kon and Lim \cite{KL20} constructed a two-server oblivious transfer protocol with quantum-key distribution and Wang et al. \cite{WKNL21-1, WKNL21-2} implemented two-server oblivious transfer protocols experimentally.

Besides the above studies, in the quantum network, it is often required to transmit quantum messages, i.e., quantum states as a subprotocol in various quantum computation tasks\cite{Wie83,GC01,Moc07,CK09,ACG+16}.
% quantum money [Wie83], quantum digital signatures [GC01], quantum coin-flipping [Moc07, CK09, ACG+16] a
However, no preceding paper studied two-server oblivious transfer for quantum messages.
In the following, this problem setting with quantum communication
is simplified to two-server quantum oblivious transfer.
Therefore, it is much demanded to develop protocols achieving this task.
In fact, the trivial method to download all states from the servers
satisfies the user secrecy condition, but does not satisfy the server secrecy condition.
Our requirement is to realize both secrecy conditions simultaneously.
This paper proposes such desired protocols.
In our proposed protocols, two servers have classical descriptions of $\sff$ quantum messages $\rho_1, \ldots, \rho_\sff$.
To implement the desired protocol, we assume that the two servers share
several entangled states.
The protocol is outlined as follows while several variants exist.
The user intends to get only one quantum states $\rho_K$ and
sends query $Q_1$ and $Q_2$ to Servers 1 and 2, respectively
while its label $K$ is not leaked to both servers.
That is, while the combination of $Q_1$ and $Q_2$ identifies
the label $K$, one query $Q_1$ nor $Q_2$ does not determine $K$.
Later, both servers send the user
their entanglement half after a certain quantum operation determined by queries $Q_1$ and $Q_2$.
Finally, the user makes a decoding operation to recover the state $\rho_K$.
When the protocol is well designed,
if the user recovers the state $\rho_K$ and the servers follow the original protocol,
the user cannot obtain any information for other quantum states $\rho_{K'}$ with
$K' \neq K$.

%\subsection{Organization of this paper}
The remainder of the paper is organized as follows.
Section \ref{S2} gives the definitions of several concepts.
Section \ref{S3} presents the outline of our protocol constructions.
Section~\ref{sec:prelim} is the technical preliminaries of the paper.
Sections \ref{SPR1} - \ref{sec:mix} are devoted for constructions of our protocols.
Section~\ref{sec:conclusion} is the conclusion of the paper.

\section{Definitions of various concepts}\Label{S2}
To briefly explain our results, we prepare 
the definitions of various concepts.
\subsection{Correctness and complexity}
%correctness, secrecy, and complexity}
To discuss the properties of our TQOT  protocols, 
we prepare several concepts.
First, we define the set $\mathcal{S}$ of possible quantum states 
as a subset of the set ${\cal S}({\cal H}_d)$ of states on ${\cal H}_d:=\mathbb{C}^d $.
A TQOT  protocol is called a TQOT  protocol over the set $\mathcal{S}$ 
when it works when the set $\mathcal{S}$ is the set of possible quantum states. 
We denote the number of messages by $\sff $. 
A TQOT  protocol $\Phi$ has two types of inputs.
The first input is $\sff$ states $(\rho_1, \ldots, \rho_\sff) \in 
\mathcal{S}^\sff$.
The second input is the choice of the label of message intended by the user, which is written as the random variable $K$.
The output of the protocol is a state $\rho_{out}$ on ${\cal H}_d$.

A TQOT  protocol $\Phi$ has bilateral communication.
The communication from the user to the servers is the upload communication,
and  
the communication from the servers to the users is the download communication.
The communication complexity is composed of 
the upload complexity and the download complexity.
The upload complexity is the sum of the communication sizes of all upload communications, and
the download complexity is the sum of the communication sizes of all 
download communications.
The sum of the upload and download complexity is called
the communication complexity.
We adopt the communication complexity
as the optimality criterion under various security conditions.

A TQOT  protocol $\Phi$ is called a deterministic protocol 
when the following two conditions hold. 
The upload complexity and the download complexity are determined only by the protocol $\Phi$.
When the user and the servers are honest,
the output is determined only by $(\rho_1, \ldots, \rho_\sff)$ and $K$.
Otherwise, it is called a probabilistic protocol.
When $\Phi$ is a deterministic protocol, we denote the output state 
by $\Phi_{out}(\rho_1, \ldots, \rho_\sff,K)= \rho_{out}$.
The upload complexity, the download complexity, and the communication complexity are 
denoted by $UC(\Phi)$, $DC(\Phi)$, and $CC(\Phi)$, respectively.
Hence, the communication complexity $CC(\Phi)$ is calculated as
$UC(\Phi)+DC(\Phi)$.

Next, we consider the case when $\Phi$ is a probabilistic protocol.
Even when the user and the servers are honest,
the user has a random variable $X$ that determines 
the upload complexity, the download complexity, and
the output state $\rho_{out}$.

We denote the distribution of $X$ by $P_{\Phi}$, and 
denote 
the upload complexity, the download complexity, 
the communication complexity, and
the output state by 
$UC_X(\Phi)$, $DC_X(\Phi)$, $CC_X(\Phi)$, and
$\Phi_{out,X}(\rho_1, \ldots, \rho_\sff,K)= \rho_{out}$, respectively.

A deterministic protocol $\Phi$ is called correct when 
%the following condition holds. When the user and the servers are honest,
the relation $\Phi_{out}(\rho_1, \ldots, \rho_\sff,\ell)=\rho_\ell$ holds
for any elements $\ell \in [\sff]$ and 
$(\rho_1, \ldots, \rho_\sff) \in \mathcal{S}^\sff$.
A probabilistic protocol $\Phi$ is called $\alpha$-correct when 
%the following condition holds. When the user and the servers are honest,
the relation $\Phi_{out}(\rho_1, \ldots, \rho_\sff,\ell)=\rho_\ell$ holds
at least probability $\alpha$
for any elements $\ell\in [\sff]$ and 
$(\rho_1, \ldots, \rho_\sff) \in \mathcal{S}^\sff$.

\subsection{User and server secrecy}
A TQOT  protocol $\Phi$ has two types of secrecy.
One is the user secrecy and the other is server secrecy.
We say that a TQOT  protocol $\Phi$ satisfies the user secrecy 
when the following condition holds. 
When the servers apply attacks and the user is honest,
no server obtains the information of the user's request $K$, i.e., 
the condition 
\begin{align}
\rho_{Y_J}=\rho_{Y_J|\ell}\label{MLK}
\end{align}
holds for any $\ell \in [\sff]$,
where $\rho_{Y_J|K}$ is the final state on Server $J$ dependently of
the variable $K$.

In contrast, we say that a TQOT  protocol $\Phi$ satisfies the server secrecy 
when the following condition holds. 
When the servers are honest and the output state $\rho_{out}$ equals $\rho_K$,
the user obtains no information for other messages $\rho_\ell$ with
$\ell \neq K$.

\subsection{Blind and visible settings}
TQOT can be studied in two distinct settings, called the {\em blind} and {\em visible} settings, 
	in which quantum state compression has also been extensively studied \cite{423,288,290,269,35,201}.
In the blind setting \cite{423,288,290,201}, 
	the servers contain 
		quantum systems $X_1,\ldots,X_{\sff}$ with the message states $\rho_1,\ldots,\rho_{\sff}$, respectively,
	but does not know the states of the systems.
Due to the no-cloning theorem, the servers cannot generate more copies of the message states and the server's operations are independent of the message states.
Each server access these systems in the encoding process.
A TQOT  protocol of the blind setting is suitable for the case where 
	the servers generate the message states by some quantum algorithm and performs the TQOT  task.

On the other hand, in the visible setting \cite{269,35,201}, 
	the servers contain the descriptions of the message states
	$\rho_1, \ldots, \rho_{\sff}$, which can be considered as continuous variables.
With the descriptions of quantum states, the servers can generate multiple copies of the quantum states, without the limitation of the no-cloning theorem, 
	and apply quantum operations depending on the descriptions of the states.
Since any protocol in the blind setting can be considered as a protocol in the visible setting, 
	we can generally expect to achieve lower communication complexity in the visible setting.
% \myred{Even if the description of a quantum state is classical information, it is infinite-length classical information,
% 	and therefore we cannot send this description via QPIR for classical messages.}
Furthermore, the visible setting is a reasonable setting for the case where the user has no ability to generate quantum states
and requires the generation of the targeted state along with the TQOT  task.

\begin{figure}[t]
\begin{center}
        \resizebox {0.9\linewidth} {!} {
\begin{tikzpicture}[scale=0.5, node distance = 3.3cm, every text node part/.style={align=center}, auto]
% \Large
	\node [rblock] (serv) {Server 1\\$(\rho_1, \ldots , \rho_{\sff})$};
	\node [rblock, below=3em of serv] (serv2) {Server 2\\$(\rho_1, \ldots , \rho_{\sff})$};
	\node [rblock, below right=0em and 12em of serv] (user) {User};
%	\node [above=3em of serv] (1) {1. prepare state\\$\rho_1\otimes \cdots \otimes \rho_{\sff}$};
% 	\node [above=3em of serv] (1) {message states\\$\rho_1, \ldots , \rho_{\sff}$};
	\node [above=3em of user] (2) {target index $\ell\in[\sff]$};
	\node [below=3em of user] (5) {retrieved state $\rho_\ell$};
	
	\path [line] (user.150) --node[above=0.1em]{query $Q$} (serv.7);
	\path [line] (serv.-7) --node[below=0.3em]{answer $A$} (user.160);
	\path [line] (user.200) --node[above]{query $Q'$} (serv2.7);
	\path [line] (serv2.-7) --node[below]{answer $A'$} (user.210);
% 	\path [line] (1) -- (serv.north);
	\path [line] (2) -- (user.north);
	\path [line] (user.south) -| (5);

% 	\path [line] (serv.22 -| user.west) --node[above]{queries\\$Q$} (serv.22);
% 	\path [line] (serv.-22) --node[below]{answers\\$A$} (serv.-22 -| user.west);
% % 	\path [line] (1) -- (serv.north);
% 	\path [line] (2) -- (user.north);
% 	\path [line] (user.south) -| (5);
	
\end{tikzpicture}
}
\caption{TQOT protocol with quantum messages.}
\label{fig:one-server}
\end{center}
\end{figure}
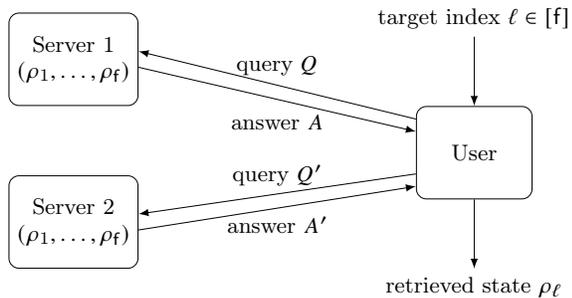

\section{Outline of obtained protocols}\Label{S3}
In this section, we propose various TQOT  protocols in the visible setting, which are summarized in Table~\ref{tab:vis}.
Our protocols prove that the TQOT  is possible in the visible setting.
In the two-server model, we assume that the servers do not communicate with each other.
	
\begin{table*}[ht]
\begin{center}
\caption{TQOT  Protocols in Visible Setting} \Label{tab:vis}
\begin{tabular}{|c|c|c|c|c|c|c|c|}
\hline
					&	\multirow{2}{*}{Message}				&	\multirow{2}{*}{Upload} 	&	\multirow{2}{*}{Download}	&	\multirow{2}{*}{Prior} & \multirow{2}{*}{User} & \multirow{2}{*}{Server}	& deterministic\\
					&	\multirow{2}{*}{States}				&	\multirow{2}{*}{Complexity}	&	\multirow{2}{*}{Complexity}	&\multirow{2}{*}{Entanglement}	& \multirow{2}{*}{Secrecy} & \multirow{2}{*}{Secrecy}	& or probabilistic\\
&&&&&&& (correctness)\\
\hline
\multirow{2}{*}{Protocol~\ref{PR1}}	&	real qubit 		&	\multirow{2}{*}{$2\sff$ bits}
	&	\multirow{2}{*}{$2$ qubits}			& \multirow{2}{*}{$1$ ebit} & \multirow{2}{*}{Yes} & \multirow{2}{*}{Yes} &	deterministic\\
&	pure states & & & & & & (correct)\\
\hline
\multirow{2}{*}{Protocol~\ref{PR2}}	&	real qudit 		&	\multirow{2}{*}{$2\sff$ bits}	
&	$2(d-1)$ 			& \multirow{2}{*}{$d-1$ ebits} & \multirow{2}{*}{Yes} & 
\multirow{2}{*}{Yes} &	deterministic\\
&	pure states & &qubits & & &	& (correct)\\
\hline
\multirow{2}{*}{Protocol~\ref{PR3}}	&	qudit commutative	&	\multirow{2}{*}{$2\sff$ bits	}
&	$2\log d$ 	& \multirow{2}{*}{$\log d$ ebit} & \multirow{2}{*}{Yes} & \multirow{2}{*}{Yes} &	deterministic\\
& unitary pure states & & qubits& & &	& (correct)\\
\hline
\multirow{3}{*}{Protocol~\ref{PR4}}&	qudit 	&	$2\sff$ bits		&
	$2(d-1) $ 	& & & &	\multirow{2}{*}{deterministic}\\
&	pure  &+ $4\log d$  & $+4\log d$&$d-1$ ebits & No &No & \multirow{2}{*}{(correct)}\\
&states&qubits& qubits &&&&\\
\hline
&	qudit&	\multirow{2}{*}{$2\sff+2d$ bits}	&	$ 2d(d-1)+$  
& $d(d-1)+$ & 
& &	\multirow{2}{*}{probabilistic}\\
Protocol~\ref{PR5} &	pure  &\multirow{2}{*}{in average} & $4 d \log d$  & $2 d\log d$ & No & Yes	& 
\multirow{2}{*}{(correct)}\\
&states	 & &  qubits in average & ebits in average &  &	 & 
\\
\hline
	&	qudit&		&	$2n(d-1)$	& $n(d-1)$  & 
& &	probabilistic\\
Protocol~\ref{PR6} & pure &$2\sff$ bits & $+4n\log d$ &$+2n\log d$ & Yes & Yes & 
\big($1-(\frac{d-1}{d})^{n}$\\
& states & & qubits &ebits&&&-correct\big)\\
\hline
	&	qudit&		&	$2n(d-1)$	& $n(d-1)$  & 
& &	probabilistic\\
Protocol~\ref{PR7}& mixed & $2\sff$ bits& $+4n\log d$ &$+(2n+1)\log d$ & Yes &	Yes & 
\big($1-(\frac{d-1}{d})^{n}$\\
&states & & qubits &ebits&&&-correct\big)\\
\hline
\end{tabular}
\end{center}
Protocols~\ref{PR4} and \ref{PR5} can be converted for mixed message states by increasing the prior entanglement by $\log d$-ebit.
\end{table*}

The goal is to construct Protocol~\ref{PR7}, i.e., a protocol that works with general mixed states in a qudit system.
This protocol is constructed by a combination of various subprotocols that work with 
the respective submodel.
Fig. \ref{F1} shows which protocol is used as a subprotocol
in each protocol construction.
In fact, when we restrict our state model into a submodel, 
we can realize much smaller download complexity.
First, we propose Protocol~\ref{PR1} that works only with real pure qubit states.
Then, using Protocol~\ref{PR1}, 
we propose Protocol~\ref{PR2} that works only with real pure qudit states.
Next, we propose Protocol \ref{PR3}, i.e., a protocol that works when
the states to be transmitted are limited to a state given by a commutative group.
These three protocols realize much smaller download complexity,
and realize the user secrecy and the server secrecy.

\begin{figure}[tb]{}
\begin{center} 
\includegraphics[width=\linewidth]{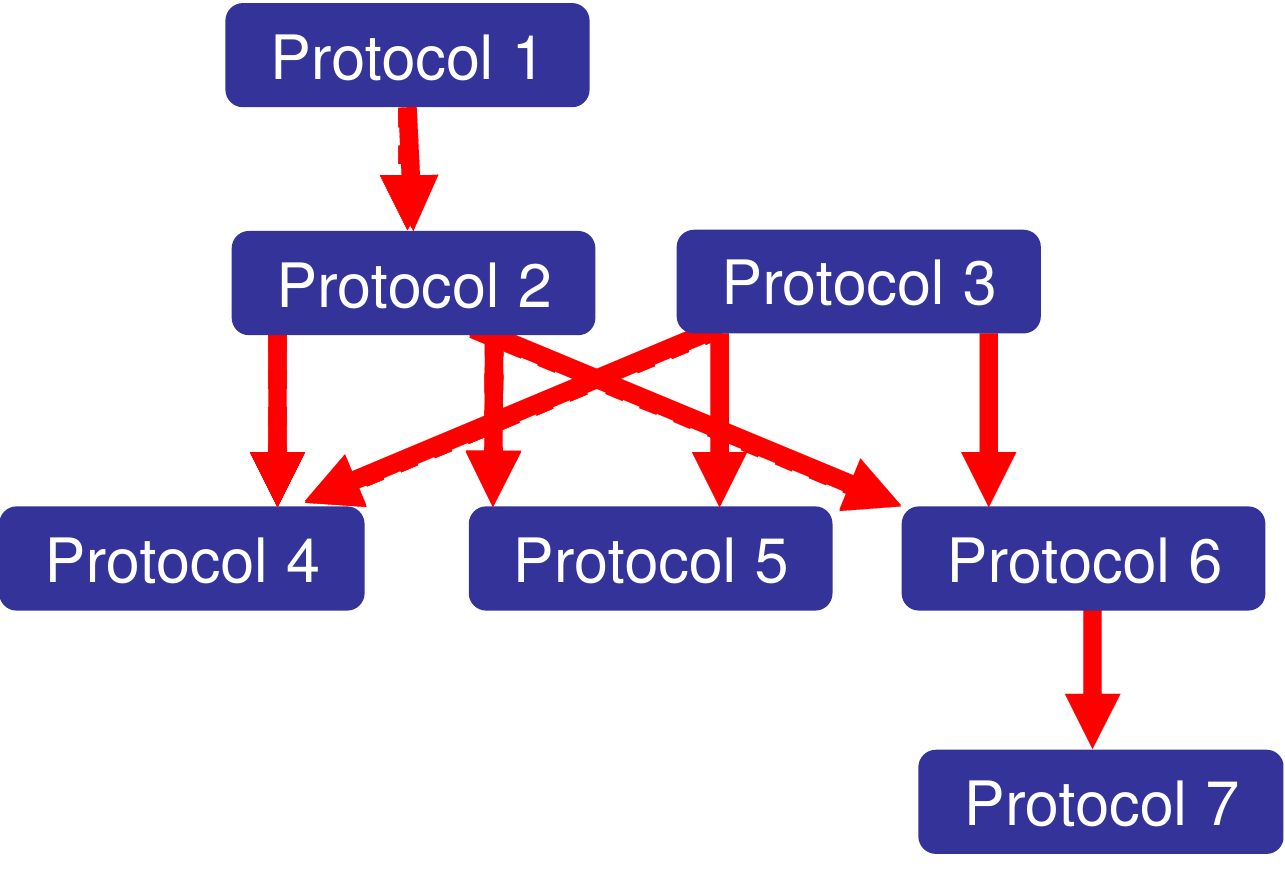}
\end{center}
\caption{{\bf Flow of protocol construction.}
Proposed protocols are constructed by using other protocol as subprotocols.
The arrow, Protocol 1 $\to$ Protocol 2, means that 
 Protocol 2 is constructed by using Protocol 1 as a subprotocol.
}\Label{F1} 
\end{figure}

Combining Protocol \ref{PR2} with a simple modification of Protocol \ref{PR3}, 
as a simple protocol, we propose Protocol \ref{PR4} that works with general pure states in a qudit system.
This protocol 
%realizes the user secrecy under the specious-server model
%and the server secrecy under the honest-user model.
%However, it 
does not have the user secrecy nor the server secrecy
under a malicious setting.
To realize the server secrecy, 
combining Protocols \ref{PR2} and \ref{PR3}, 
we propose Protocol \ref{PR5} that works with general pure states in a qudit system.
This protocol realizes the server secrecy.
However, it does not have the user secrecy.
To realize the user secrecy, 
combining Protocols \ref{PR2} and \ref{PR3} in a way different from Protocol \ref{PR5},
we propose Protocol \ref{PR6} that works with general pure states in a qudit system.
This protocol realizes 
the user secrecy and the server secrecy.
To adopt mixed states,
modifying Protocol \ref{PR6},
we propose Protocol \ref{PR7} that works with general mixed states in a qudit system.

\section{Preliminaries} \Label{sec:prelim}
We define $[a:b]  = \{a,a+1, \ldots, b\}$ and $[a] = \{1,\ldots, a\}$.
%Given quantum systems $H_1,\ldots, H_n$, quantum states $\rho_1,\ldots, \rho_n$, quantum operations $\Lambda_1,\ldots,\Lambda_n$, and a set $S \subset [n]$,
%	we denote $H_{S} \coloneqq \bigotimes_{i\in S} H_i$,
%	$\rho_{S} \coloneqq \bigotimes_{i\in S} \rho_i$,
%	and
%	$\Lambda_{S} \coloneqq \bigotimes_{i\in S} \Lambda_i$.
The dimension of a quantum system $X$ is denoted by $|X|$. 

Throughout this paper, 
$\mathbb{C}^d$ expresses the $d$-dimensional Hilbert space
spanned by the orthogonal basis $\{|s\rangle\}_{s=0}^{d-1}$.
For a $d_1\times d_2$ matrix 
	\begin{align}
	\sM = \sum_{s=0}^{d_1-1} \sum_{t=0}^{d_2-1} m_{st} |s\rangle\langle t|  \in \mathbb{C}^{d_1\times d_2},
	\end{align}
	we define 
\begin{align}
|\sM\rrangle =
\frac{1}{\sqrt{d}}\sum_{s=0}^{d_1-1} \sum_{t=0}^{d_2-1} m_{st} |s\rangle| t\rangle  \in \mathbb{C}^{d_1}\otimes \mathbb{C}^{d_2}.
\end{align}
For $\sA \in \mathbb{C}^{d_1\times d_2}$, $\sB\in\mathbb{C}^{d_1 \times d_1}$, and $\sC \in\mathbb{C}^{d_2 \times d_2}$,
	we have the relation %holds:
\begin{align}
(\sB\otimes \sC^{\top}) | \sA \rrangle =  | \sB\sA\sC \rrangle.
\end{align}

%For the description of the protocol, we define the generalized Pauli operators and maximally entangled state for $d$-dimensional systems by

We call a $d$-dimensional system $\mathbb{C}^d$ a {\em qudit}.
%	and a $|\sI\rrangle\in\mathbb{C}^d\otimes \mathbb{C}^d$ an {\em edit}.
Define generalized Pauli matrices and the maximally entangled state on qudits as 
\begin{align}
\sX_d &= \sum_{s=0}^{d-1} |s+1\rangle \langle s|,\\ 
\sZ_d &= \sum_{s=0}^{d-1} \omega^{s} |s\rangle \langle s|,\\
%\sY_d &= \sX_d \sZ_d,\\
|\sI_d \rrangle &= \frac{1}{\sqrt{d}} \sum_{s=0}^{d-1} |s,s\rangle,
\Label{eq:defs-pauli}
\end{align}
where $\omega = \exp(2\pi\ii/ d)$ and $\iota = \sqrt{-1}$.
%{Notice that we define $\sY_2$ as $\sX_2\sZ_2$ not $\ii \sX_2\sZ_2$ for our convenience in later sections.}
%
%Define generalized Pauli matrices as
%\begin{align}
%X &= \sum_{i=0}^{d-1} |i+1\rangle \langle i|, \\
%Z &= \sum_{i=0}^{d-1} \omega^i |i\rangle\langle i|, \\
%Y &= {\omega XZ,}
%\end{align}
%where $\omega = \exp(2\pi\ii/ d)$.
We define 
	the generalized Bell measurements
	\begin{align}
		\mathbf{M}_{\sX\sZ,d} = \{ |\sX^a \sZ^b\rrangle \mid a,b \in [0:d-1] \}. \Label{mes}
		%\mathbf{M}_{\sZ\sY,d} = \{ |\sZ^a \sY^b\rrangle \mid a,b \in [0:d-1] \}
	\end{align}
If there is no confusion, we denote $\sX_d,\sZ_d, \sI_d, \mathbf{M}_{\sX\sZ,d}$ by $\sX,\sZ, \sI, \mathbf{M}_{\sX\sZ}$.
Let $A, A',B, B'$ be qudits.
If the state on $A\otimes A' \otimes B \otimes B'$ is $|\sA\rrangle\otimes|\sB\rrangle$ %$|I\rrangle\otimes |I\rrangle$
	and the measurement $\mathbf{M}_{\sX\sZ}$
	is performed on $A' \otimes B'$ with outcome $(a,b) \in [0:d-1]^2$,
	the resultant state is 
	%\begin{align}
	%|Z^{a}Y^{b} \rrangle \in A \otimes B.
	%\end{align}
%If the state on $A\otimes A' \otimes B \otimes B'$ is $|A\rrangle\otimes|B\rrangle$ instead of  $|I\rrangle\otimes |I\rrangle$,
%	the resultant state for the measurement outcome $(a,b)$ is
	\begin{align}
	|\sA\sX^a\sZ^{-b} \sB^\top\rrangle \in A \otimes B.
		\Label{qe:feafterf}
	\end{align}
Also, we define the unitary $\sV$ on
$\mathbb{C}^d\otimes \mathbb{C}^d$ as
\begin{align}
\sV|j\rangle |j'\rangle =|j\rangle |j'+j\rangle,
\end{align}
which implies the relation $\sV|j\rangle |0\rangle=|j\rangle |j\rangle$.
We define the following state
\begin{align}
|+\rangle:= \frac{1}{\sqrt{d}}\sum_{j=0}^{d-1}|j\rangle \in \mathbb{C}^d.
\end{align}

\section{Symmetric QPIR protocol for pure real qubit states}\Label{SPR1}
In this subsection, we construct a two-server TQOT  protocol for pure qubit states in the visible setting.
%Define the following unitaries
Define the rotation operation on $\mathbb{C}^2$ and the phase-shift operation by 
\begin{align*}
	\Rr(\theta) := 
	\begin{pmatrix}
	\cos\theta & - \sin\theta	\\
	\sin\theta & \cos\theta
	\end{pmatrix}
	,\quad
	\Ss(\varphi) := 
	\begin{pmatrix}
	e^{-\ii\varphi/2}	&	0	\\
	0	&	e^{\ii\varphi/2}
	\end{pmatrix}
\end{align*}
for $\theta, \varphi \in [0,2\pi)$.
% $R_{\theta}$ and $\Rr(\theta')$ are commutative, respectively. 
For any $\varphi,\varphi', \theta,\theta'$,  we have 
\begin{align}
\Rr(\theta) \Rr(\theta') = \Rr(\theta+\theta'), \quad 
\Ss(\varphi) \Ss(\varphi') = \Ss(\varphi+\varphi'),
\end{align}
and therefore, 
	$\Ss(\varphi)$ and $\Ss(\varphi')$ ($\Rr(\theta)$ and $\Rr(\theta')$) are commutative.
Also, we have
\begin{align}
|\Rr(\theta)\rrangle= 
\cos \theta (|00\rangle +|1,1\rangle)
+\sin \theta (|10\rangle -|0,1\rangle)\Label{MZE}.
\end{align}
We also define the unitary $\sT$ on $\mathbb{C}^2\otimes \mathbb{C}^2$ as
		\begin{align} 
		\sT(\frac{1}{\sqrt{2}}(|0,0\rangle+|1,1\rangle)= |0,0\rangle,~
		\sT(\frac{1}{\sqrt{2}}(|1,0\rangle-|0,1\rangle)= |1,0\rangle\\
		\sT(\frac{1}{\sqrt{2}}(|0,0\rangle-|1,1\rangle)= |0,1\rangle,~
		\sT(\frac{1}{\sqrt{2}}(|1,0\rangle+|0,1\rangle)= |1,1\rangle.
		\end{align}
Then, using \eqref{MZE}, we have
\begin{align}
\sT|\Rr(\theta)\rrangle
=\cos \theta |00\rangle+\sin \theta |10\rangle
=(\sR(\theta) |0\rangle)|0\rangle.\Label{ZSY}
\end{align}

We also have 
\begin{align}
\Rr(\theta)^{\top} = \Rr(-\theta)
,\quad 
\Ss(\varphi)^{\top} = \Ss(\varphi) 
\end{align}
and
\begin{align}
\sX\Rr(\theta)\sX = \Rr(-\theta),\quad 
\sX \Ss(\varphi) \sX = \Ss(-\varphi) 
.
\Label{eq:xrxxsx}
\end{align}
As special cases, we have $\sY \coloneqq \sX\sZ = \Rr(\pi/2)$ and $\sZ = \Ss(\pi)$.

Any real vector can be written as 
\begin{align}
 \left(
\begin{array}{c}
\cos \theta \\
\sin \theta 
\end{array}
 \right)=
\Rr(\theta)|0\rangle.
%=\Rr(\theta+\frac{\pi}{2})|1\rangle.
 \end{align}

Now, we construct a TQOT  protocol in the visible setting for real qubit states
%, which achieves the communication complexity in Theorem~\ref{theo:31}.
\begin{prot}[TQOT   protocol for real qubit pure states] \Label{PR1}
% The message states are defined in \eqref{eq:ins}
% 	and 
For any message 
real qubit states $|\psi_1\rangle,\ldots, |\psi_{\sff}\rangle \in \pureset{\mathbb{R}^2}$,
	we choose the parameters $\theta_\ell$ as
	%where each state $|\psi_\ell\rangle$ is written as 
\begin{align}	
	|\psi_\ell\rangle = \Rr(\theta_\ell)|0\rangle.
\end{align}
	When the user's target index $K$ is $k\in[\sff]$, i.e., the targeted state is $|\psi_k\rangle$, our protocol is given as follows.
\begin{description}[leftmargin=1.5em]
\item[0)] \textbf{Entanglement Sharing}: 
	Let $A, A'$ be qubits. Before starting the protocol,
	Server 1 and Server 2 share 
				a maximally entangled state $|\sI_2\rrangle$ on $A\otimes A'$,
%	Server 1 and Server 2 share 
%				two maximally entangled state $|I\rrangle$ on $A\otimes A'$ and $B\otimes B'$,
				where Server 1 (Server 2) contains $A$ ($A'$).
\item[1)] \textbf{Query}: %The same as Protocol \ref{PR1}.
			The user chooses 
				$Q= (Q_{1},\ldots, Q_{\sff}) \in  \{0,1\}^{\sff}$ uniformly at random.
			The variable $Q' = (Q_{1}',\ldots,Q_{\sff}' )\in  \{0,1\}^{\sff}$ is defined as 
				\begin{align}
				Q_{\ell}' = 
					\begin{cases}  
					Q_{\ell} & \text{for $\ell\neq k$},	\\ 
					Q_{\ell}\oplus 1 & \text{for $\ell=k$}. 
					\end{cases}\Label{AMY}
				\end{align}
			The user sends $Q$ and $Q'$ to Server 1 and Server 2, respectively.
\item[2)] \textbf{Answer}: 
				{When $Q=q$ and $Q'=q'$,} 
				Server $1$ applies $\Rr( \sum_{\ell=1}^\sff q_\ell \theta_\ell )$
				on $A$,
				and sends $A$ to the user.
				Similarly, Server $2$ applies $\Rr( \sum_{\ell=1}^\sff q_\ell' \theta_\ell )$
				on $A'$,
				and sends $A'$ to the user.
				
% 				Server $2$ sends the states
% 				\begin{align}
% 				|\phi_2\rangle 
% 				&= |(U_\sff^{\dagger})^{q_{2,\sff}} \cdots (U_{1}^{\dagger})^{q_{2,1}} \rrangle 
% 				= (I_d\otimes \bar{U}_1^{q_{2,1}} \cdots \bar{U}_{\sff}^{q_{2,\sff}}) |I_d \rrangle
% 					\in A_{2} \otimes A_{2}'  = \mathbb{C}^{d}\otimes\mathbb{C}^{d} \\
% 				|\psi_2\rangle 
% 				&= |(U_\sff^{\dagger})^{r_{2,\sff}} \cdots (U_{1}^{\dagger})^{r_{2,1}}  \rrangle 
% 				= (I_d\otimes \bar{U}_1^{r_{2,1}} \cdots \bar{U}_{\sff}^{r_{2,\sff}}) |I_d \rrangle
% 					\in B_{2} \otimes B_{2}'  = \mathbb{C}^{d}\otimes\mathbb{C}^{d} 
% 				\end{align}
% 				to the user.

\item[3)] \textbf{Reconstruction}:
	When both servers are honest, the user receives the state 
	$|\Rr( \sum_{\ell=1}^\sff (q_\ell-q_\ell') \theta_\ell )\rrangle
	=	|\Rr( (-1)^{q_k+1} \theta_k )\rrangle$ on $A\otimes A'$.
	\begin{enumerate}
		%\item[a)]
		\item
		The user applies the unitary $\sT$ on $A\otimes A'$.
		Then, the resultant state is 
		$( \Rr( (-1)^{q_k+1} \theta_k )|0 \rangle )|0 \rangle
		= (\sZ^{q_k+1}
		\Rr( \theta_k )|0 \rangle) |0 \rangle$ due to \eqref{ZSY}.
		\item
				The user traces out $A'$ and applies $\sZ^{q_k+1}$ on $A$.
Then, the resultant stare on $A$ is $\Rr(  \theta_k)|0 \rangle$. 
\if0		
				The user performs the measurement with the basis $\{ |0\rangle,|1\rangle \}$ 
				on $A'$, and obtain the outcome $s$.
				Then, the user's state on $A$ is 
				$\Rr( (-1)^{q_k+1} \theta_k) |s \rangle
				=\Rr( (-1)^{q_k+1} \theta_k +\frac{\pi}{2})s|0 \rangle
				=\Rr( \frac{\pi}{2}s) \Rr( (-1)^{q_k+1} \theta_k)|0 \rangle
				=\Rr( \frac{\pi}{2}s) \sZ^{q_k+1} \Rr(  \theta_k)|0 \rangle$.
		\item
				The user applies $\sZ^{q_k+1}\Rr( -\frac{\pi}{2}s) $ on $A$.
Then, the resultant stare on $A$ is $\Rr(  \theta_k)|0 \rangle$. 
		\fi		\QEDA
	\end{enumerate}
\end{description}
\end{prot}

Protocol~\ref{PR1} satisfies the correctness, secrecy, and communication complexity, which is shown as follows.

\begin{lemm} \Label{LPR1}
Protocol \ref{PR1} is a correct TQOT  protocol that satisfies 
the user secrecy and the server secrecy. 
Its upload complexity and its download complexity
are $2\sff$ bits and $2$ qubits, respectively.
The required prior entanglement is one copy of $|I_2\rrangle$, i.e., one ebit.
\end{lemm}
\begin{proof}
The correctness and the complexity are
 shown during the protocol description.
The secrecy can be shown as follows.
Throughout the protocol, the servers only obtain the queries, and 
	each query is uniformly random $\sff$ bits.
Therefore, each server does not obtain any information of $k$.
Hence, the user secrecy holds.
On the other hand, at the end of the step of answer, the user obtains the state 
$|\Rr( \sum_{\ell=1}^\sff (q_\ell -q_\ell')\theta_\ell )\rrangle=
\sT^{\dagger} 
\Rr( \sum_{\ell=1}^\sff (q_\ell -q_\ell')\theta_\ell )|0\rangle
|0\rangle$.
%That is, the user receives only information that is unitarily equivalent to 
%the one-qubit information
%$\Rr( \sum_{\ell=1}^\sff (q_\ell -q_\ell')\theta_\ell )|0\rangle$.
That is, for any malicious queries $q, q'$, 
the state depends only on $\sum_{\ell=1}^\sff (q_\ell -q_\ell')\theta_\ell$.
In order to recover the $\Rr(\theta_k)|0\rangle$,
$\sum_{\ell=1}^\sff (q_\ell -q_\ell')\theta_\ell$ needs to have a one-to-one relation to $\theta_k$.
Hence, when the user recovers the quantum message $\Rr(\theta_k)|0\rangle$, he can obtain 
 no information for other $\theta_\ell$.
Hence, the server secrecy holds.
\end{proof}

\section{Symmetric QPIR protocol for pure real qudit states}
To construct a TQOT  protocol for real qudit states,
we first consider the parameterization of pure real states on $d$-dimensional systems.
Define $d \times d$ matrix $R(\theta^1,..,\theta^{d-1} )$ as
\begin{align}
\Rr(\theta^1,\ldots, \theta^{d-1})
 &= \Rr_{d-1}(\theta^1) \cdots \Rr_{1}(\theta^{d-1}),
\end{align}
where $\Rr_{s}(\theta)$ is the rotation 
	\begin{align}
	\begin{pmatrix}
	\cos\theta & - \sin\theta	\\
	\sin\theta & \cos\theta
	\end{pmatrix}
	\end{align}
	with respect to the two basis elements $|s-1\rangle$ and $|s\rangle$.
Notice that 
any two of $\Rr_{1}(\theta^1), \ldots, \Rr_{d-1}(\theta^{d-1})$ are not in general.
We also have
	$\Rr_{s}(\theta^s)^{\top} = \Rr_{s}(-\theta^s)$.
Any real vector can be written as 
\begin{align}
\Rr(\theta^1,\ldots, \theta^{d-1})|0\rangle.
 \end{align}

Now, we construct a TQOT  protocol in the visible setting for real qudit states
%, which achieves the communication complexity in Theorem~\ref{theo:31}.
\begin{prot}[TQOT   protocol for real qubit pure states] \Label{PR2}
% The message states are defined in \eqref{eq:ins}
% 	and 
For any message real pure states 
$|\psi_1\rangle,\ldots, |\psi_{\sff}\rangle \in \pureset{\mathbb{R}^d}$,
	we choose the parameters $\theta_\ell^1,\ldots, \theta_\ell^{d-1}$ as
	%where each state $|\psi_\ell\rangle$ is written as 
\begin{align}	
	|\psi_\ell\rangle = 
		\Rr(\theta^1_{\ell},\ldots, \theta^{d-1}_{\ell})
		|0\rangle.
\end{align}
When the user's target index $K$ is $k\in[\sff]$, i.e., the targeted state is $|\psi_k\rangle$, our protocol is given as follows.
\begin{description}[leftmargin=1.5em]
\item[0)] \textbf{Entanglement Sharing}: 
	Let $A_1,\ldots, A_{d-1} , A_1',\ldots, A_{d-1}'$ be qubits. 
	Before starting the protocol,
	Server 1 and Server 2 share 
				$d-1$ maximally entangled state $|\sI_2\rrangle$ on 
				$A_1\otimes A_1', \ldots, A_{d-1}\otimes A_{d-1}'$,
%	Server 1 and Server 2 share 
%				two maximally entangled state $|I\rrangle$ on $A\otimes A'$ and $B\otimes B'$,
				where Server 1 (Server 2) contains $A_1,\ldots, A_{d-1}$ ($A_1',\ldots, A_{d-1}'$).
\item[1)] \textbf{Query}: 
			The same as Protocol~\ref{PR1}.
\item[2)] \textbf{Answer}: 
				{When $Q=q$ and $Q'=q'$,} 
				Server $1$ applies $\Rr( \sum_{\ell=1}^\sff q_\ell \theta_\ell^j )$
				on $A_j$,
				and sends $A_1,\ldots, A_{d-1}$ to the user.
				Similarly, Server $2$ applies $\Rr( \sum_{\ell=1}^\sff q_\ell' \theta_\ell^j )$
				on $A_j'$ for $j=1,\ldots, d-1$,
				and sends $A_1',\ldots, A_{d-1}'$ to the user.
				
% 				Server $2$ sends the states
% 				\begin{align}
% 				|\phi_2\rangle 
% 				&= |(U_\sff^{\dagger})^{q_{2,\sff}} \cdots (U_{1}^{\dagger})^{q_{2,1}} \rrangle 
% 				= (I_d\otimes \bar{U}_1^{q_{2,1}} \cdots \bar{U}_{\sff}^{q_{2,\sff}}) |I_d \rrangle
% 					\in A_{2} \otimes A_{2}'  = \mathbb{C}^{d}\otimes\mathbb{C}^{d} \\
% 				|\psi_2\rangle 
% 				&= |(U_\sff^{\dagger})^{r_{2,\sff}} \cdots (U_{1}^{\dagger})^{r_{2,1}}  \rrangle 
% 				= (I_d\otimes \bar{U}_1^{r_{2,1}} \cdots \bar{U}_{\sff}^{r_{2,\sff}}) |I_d \rrangle
% 					\in B_{2} \otimes B_{2}'  = \mathbb{C}^{d}\otimes\mathbb{C}^{d} 
% 				\end{align}
% 				to the user.

\item[3)] \textbf{Reconstruction}:
	When both servers are honest, 
	the user receives 	$|\Rr( (-1)^{q_k+1} \theta_k^1 )\rrangle, \ldots, 
	|\Rr( (-1)^{q_k+1} \theta_k^{d-1} )\rrangle$. 
	\begin{enumerate}
		%\item[a)]
		\item
	Applying the procedure given in Step 3) of Protocol~\ref{PR1}, 
	the user constructs states 
	$	\Rr( \theta_k^1 )|0 \rangle \otimes \cdots \otimes\Rr( \theta_k^{d-1} )|0 \rangle$ 
	on $A_1\otimes \cdots \otimes A_{d-1}$
	from the received states. 
		\item
				The user applies Kraus operators $\{ \sF_{1,1},\sF_{1,2}\} $ 
				with input system $A_1 \otimes A_2$ and the output system ${\cal H}_3$, where
				\begin{align}
				\sF_{1,1}&:=\frac{1+\iota}{\sqrt{2}}|0\rangle \langle 0| \langle u_+|
				     + \frac{1}{\sqrt{2}}(|1\rangle\langle 1| \langle 0|+|2\rangle\langle 1| \langle 1|)\\
				\sF_{1,2}&:=\frac{1+\iota}{\sqrt{2}}|0\rangle \langle 0| \langle u_-|
				     + \frac{1}{\sqrt{2}}(|1\rangle\langle 1| \langle 0|+|2\rangle\langle 1| \langle 1|)
\end{align}
and $|u_{\pm}\rangle:= \frac{1}{\sqrt{2}}(|0\rangle\pm \iota|1 \rangle )$.
	When both servers are honest, the resultant state is 
	$ \Rr(\theta^1_k, \theta^{2}_k)|0\rangle$.
	
		The user makes the following procedure inductively for $j=1, \ldots, d-2$: 
\item[j+1)] 
		The user applies Kraus operators $\{ \sF_{j,1},\sF_{j,2}\} $ with input system 
${\cal H}_{j+1} \otimes A_{j+1}$ and the output system ${\cal H}_{j+2}$, where
				\begin{align}
				\sF_{j,1}&:=\frac{1+\iota}{\sqrt{2}} \sum_{j'=0}^{j-1}|j'\rangle \langle j'| \langle u_+|
				     + \frac{1}{\sqrt{2}}(|j\rangle\langle j| \langle 0|+|j+1\rangle\langle j| \langle 1|)\\
				\sF_{j,2}&:=\frac{1+\iota}{\sqrt{2}}\sum_{j'=0}^{j-1}|j'\rangle \langle j'| \langle u_-|
				     + \frac{1}{\sqrt{2}}(|j\rangle\langle j| \langle 0|+|j+1\rangle\langle j| \langle 1|).
\end{align}
	When both servers are honest, the resultant state is 
	$ \Rr(\theta^1_k, \ldots,\theta^{j+1}_k)|0\rangle\in {\cal H}_{j+2}$. 
	\end{enumerate}
\end{description}
\end{prot}

Protocol~\ref{PR2} satisfies the correctness, secrecy, and communication complexity, which is shown as follows.

\begin{lemm} \Label{LPR2}
Protocol \ref{PR2} is a correct TQOT  protocol that satisfies 
the user secrecy and the server secrecy. 
Its upload complexity and its download complexity
are $2\sff$ bits and $2(d-1)$ qubits, respectively.
The required prior entanglement is $d-1$ copies of $|I_2\rrangle$, i.e., $d-1$ ebits.
\end{lemm}
\begin{proof}
The correctness and the complexity shown during the protocol description.
The secrecy can be shown as follows.
Throughout the protocol, the servers only obtain the queries, and 
	each query is uniformly random $\sff$ bits.
Therefore, each server does not obtain any information of $k$.
Hence, the user secrecy holds.

When both servers are honest and the user is malicious,
at the end of the step of answer, the user obtains the states 
$|\Rr( \sum_{\ell=1}^\sff (q_\ell -q_\ell')\theta_\ell^{j} )\rrangle$ with $j=1, \ldots, d-1$.
In order that the user recovers the state 	$ \Rr(\theta^1_k, \ldots,\theta^{j+1}_k)|0\rangle$,
the states 
$|\Rr( \sum_{\ell=1}^\sff (q_\ell -q_\ell')\theta_\ell^{j} )\rrangle$ with $j=1, \ldots, d-1$
need to reflect the parameters $\theta_k^{1},\ldots, \theta_k^{d-1}$.
That is, it is needed to recover $\theta_k^j$
from $\sum_{\ell=1}^\sff (q_\ell -q_\ell')\theta_\ell^{j}$ for $j=1, \ldots, d-1$.
Hence, $\sum_{\ell=1}^\sff (q_\ell -q_\ell')\theta_\ell^{j}$ 
has one-ton-one relation to $\theta_k^j$ for $j=1, \ldots, d-1$.
Thus, $\sum_{\ell=1}^\sff (q_\ell -q_\ell')\theta_\ell^{j}$ 
has no information for other $\theta_\ell^j$ for $j=1, \ldots, d-1$.
Then, the user has no information for other $\theta_\ell$.
Hence, the server secrecy holds.
\end{proof}

\section{TQOT protocol for pure states described by commutative unitaries}\Label{S-PR3}

%\subsubsection{QPIR Protocol with prior entangelment}
%When prior entanglement between servers is allowed, 
%	Protocol~\ref{prot:comm} can be modified as follows.
In this subsection, we construct a two-server TQOT  protocol 
for pure states described by commutative unitaries in the visible setting.

\begin{prot}[TQOT protocol for pure states described by commutative unitaries] \Label{PR3}
%The message states are defined in \eqref{eq:ins}
	For commutative $\sff$ unitaries $\sU_1, \ldots, \sU_\sff$ on $\mathbb{C}^d$, the message states are given as
	\begin{align}
	|\sU_1\rrangle, \ldots ,|\sU_\sff\rrangle	\in \mathbb{C}^{d}\otimes \mathbb{C}^{d}.	\Label{eq:ins}
	\end{align}
%The user's target index is $k\in[\sff]$, i.e., the targeted state is $|U_k\rrangle$.
	When the user's target index $K$ is $k\in[\sff]$, i.e., the targeted state is $|\sU_k\rrangle$, our protocol is given as follows.
\begin{enumerate}[leftmargin=1.5em]
\item \textbf{Entanglement Sharing}: 
	Let $A, A', B, B'$ be qudits.
	Server 1 and Server 2 share 
				two maximally entangled states $|\sI_d\rrangle$ on $A\otimes A'$ and $B\otimes B'$,
%	Server 1 and Server 2 share 
%				two maximally entangled state $|I\rrangle$ on $A\otimes A'$ and $B\otimes B'$,
				where Server 1 (Server 2) contains $A\otimes B$ ($A' \otimes B'$).
\item \textbf{Query}: 
The same as Protocol \ref{PR1}.
\item \textbf{Answer}: 
				{When $Q=q$ and $Q'=q'$,} 
				Server $1$ applies 
				\begin{align}
				&\sU_1^{q_{1}} \cdots \sU_{\sff}^{q_{\sff}},\\
				&\sU_1^{-q_{1}} \cdots \sU_{\sff}^{-q_{\sff}}
				\end{align}
				on $A$ and $B$, respectively,
				and sends $A$ and $B$ to the user.
				Similarly, Server $2$ applies 
				\begin{align}
				&\bar{\sU}_1^{q_{1}'} \cdots \bar{\sU}_{\sff}^{q_{\sff}'},\\
				&\bar{\sU}_1^{-q_{1}'} \cdots \bar{\sU}_{\sff}^{-q_{\sff}'}
				\end{align}
				on $A'$ and $B'$, respectively,
				and sends $A'$ and $B'$ to the user.
				
% 				Server $2$ sends the states
% 				\begin{align}
% 				|\phi_2\rangle 
% 				&= |(U_\sff^{\dagger})^{q_{2,\sff}} \cdots (U_{1}^{\dagger})^{q_{2,1}} \rrangle 
% 				= (I_d\otimes \bar{U}_1^{q_{2,1}} \cdots \bar{U}_{\sff}^{q_{2,\sff}}) |I_d \rrangle
% 					\in A_{2} \otimes A_{2}'  = \mathbb{C}^{d}\otimes\mathbb{C}^{d} \\
% 				|\psi_2\rangle 
% 				&= |(U_\sff^{\dagger})^{r_{2,\sff}} \cdots (U_{1}^{\dagger})^{r_{2,1}}  \rrangle 
% 				= (I_d\otimes \bar{U}_1^{r_{2,1}} \cdots \bar{U}_{\sff}^{r_{2,\sff}}) |I_d \rrangle
% 					\in B_{2} \otimes B_{2}'  = \mathbb{C}^{d}\otimes\mathbb{C}^{d} 
% 				\end{align}
% 				to the user.

\item \textbf{Reconstruction}:
				The user outputs the state on $A \otimes A'$ if $Q_{k} = 1$, otherwise outputs the state on $B\otimes B'$. \QEDA
\end{enumerate}
\end{prot}

%\textbf{Correctness:}
%By the same derivation as \eqref{eq:commustaa},
%	the user obtains $|U_k\rrangle$.
%

%The correctness, secrecy, and communication complexity of Protocol~\ref{prot:wpe} are analyzed as follows.
Protocol~\ref{PR3} satisfies the correctness, secrecy, and communication complexity, which is shown as follows.

\begin{lemm} \Label{LPR3}
Protocol \ref{PR3} is a correct TQOT  protocol that satisfies 
the user secrecy and the server secrecy. 
Its upload complexity and its download complexity
are $2\sff$ bits and $2 \log d$ qubits, respectively.
The required prior entanglement is one copy of $|I_d\rrangle$, i.e., 
$\log d$ ebits.
\end{lemm}
\begin{proof}
%The correctness and the complexity shown during the protocol description.
To consider correctness,
we consider the case where $Q_{k} =1$.
%If the measurement outcome is $|I\rrangle$,
When $Q=q$, the state on $A\otimes A'$ after the measurement is 
% 	\begin{align}
% 	(U_1^{q_{1,1}} \cdots U_{\sff}^{q_{1,\sff}} \otimes I_d \otimes I_d\otimes \bar{U}_1^{q_{2,1}} \cdots \bar{U}_{\sff}^{q_{2,\sff}}) |I_d \rrangle
% 	 |I_d \rrangle
% 	\end{align}
	\begin{align}
	&(\sU_1^{q_{1}} \cdots \sU_{\sff}^{q_{\sff}} \otimes \bar{\sU}_1^{q_{1}'} \cdots \bar{\sU}_{\sff}^{q_{\sff}'}) |\sI_d \rrangle\nonumber \\
	&=(\sU_1^{q_{1}} \cdots \sU_{\sff}^{q_{\sff}} \otimes \bar{\sU}_1^{q_{1}'} \cdots \bar{\sU}_{\sff}^{q_{\sff}'}) |\sI_d \rrangle\nonumber \\
	&=|\sU_1^{q_{1}} \cdots \sU_{\sff}^{q_{\sff}} (\sU_{\sff}^{\dagger})^{q_{\sff}'} \cdots ({\sU}_1^{\dagger})^{q_{1}'} \rrangle\nonumber \\
	%&=|U_k U_1^{q_{1}}({U}_1^{q_{1}})^{\dagger}  \cdots U_{\sff}^{q_{\sff}} (U_{\sff}^{q_{\sff}})^{\dagger} \rrangle \Label{eq:commust}\\
	&=|\sU_k \rrangle, \Label{eq:commust} 
% 	=&|U_k^{ \rrangle
	\end{align}
where \eqref{eq:commust} follows from the commutativity of the unitaries $\sU_1,\ldots, \sU_{\sff}$,
	$q_{\ell}\oplus q_{\ell}' = \delta_{\ell,k}$, and $(q_{k},q_{k}') = (1,0)$.
By similar analysis, if $Q_{k}= 0$, the resultant state on $B\otimes B'$ is $|\sU_k\rrangle$. 
	%when the measurement outcome is $|I\rrangle$.

%The complexity is calculated as follows.
%The query is $2\sff$ bits, the communication from the server to the user is $4$ qudits, i.e., $4\log d$ qubits,
%and prior entanglement is $2$ copies of $|I_d\rrangle$, i.e., $2\log d$ ebits.

The secrecy can be shown as follows.
Since the queries are the same as Protocol \ref{PR1}, 
the user secrecy holds.
%Throughout the protocol, the servers only obtain the queries, and 
%	each query is uniformly random $\sff$ bits.
%Therefore, each server does not obtain any information of $k$.

When both servers are honest and the user is malicious,
at the end of the protocol, the user obtains both of 
$|\prod_{\ell=1}^{\sff} \sU_\ell^{q_\ell-q_\ell'} \rrangle$ and 
$|\prod_{\ell=1}^{\sff} \sU_\ell^{-q_\ell+q_\ell'}\rrangle$.
To order that the user recovers $|\sU_k\rrangle$, 
the state $|\prod_{\ell=1}^{\sff} \sU_\ell^{q_\ell-q_\ell'} \rrangle$ or
$|\prod_{\ell=1}^{\sff} \sU_\ell^{-q_\ell+q_\ell'}\rrangle$
needs to have one-to-one relation to $|\sU_k \rrangle$.
To realize this relation, $q_\ell$ equals $q_\ell'$ for $\ell \neq k$.
%Although $|\sU_k^{\dagger}\rrangle$ is transmitted additionally, 
In this case, both states have no information for all other message states.
Hence, no information for all other message states is leaked to the user.
\end{proof}

\section{Deterministic TQOT  protocol for pure qudit states}
We first consider the parameterization of pure states on $d$-dimensional systems.
	%and construct a QPIR protocol.
%\subsubsection{Parameterization of pure states by unitary matrices}
{Define $d \times d$ matrix $S(\varphi^1,..,\varphi^{d-1} )$ as}
\begin{align}
\Ss(\varphi^1,\ldots, \varphi^{d-1})
 &= |0\rangle\langle 0| + \sum_{s=1}^{d-1} e^{\ii \varphi^s} |s\rangle\langle s|\nonumber \\
 &= \begin{pmatrix}
 1 & 0 & 0 & 0 \\
 0 & e^{\ii \varphi^1} & 0 & 0 \\
 0 & 0 & \ddots & 0 \\
 0 & 0& 0& e^{\ii \varphi^{d-1}}
 \end{pmatrix}.
\end{align}
Notice that 
	$\Ss(\varphi^1,\ldots, \varphi^{d-1} )$ for all $\varphi^1,\ldots,\varphi^{d-1}$ 
	are commutative.
We also have
$\Ss(\varphi^1,\ldots, \varphi^{d-1} )^\top = \Ss(\varphi^1,\ldots, \varphi^{d-1} )$.
It can be easily checked that any pure state $|\psi\rangle\in\pureset{\mathbb{C}^d}$ is written 
	in the form
\begin{align}
|\psi\rangle = 
		\Ss(\varphi^1,\ldots, \varphi^{d-1})
		\Rr(\theta^1,\ldots, \theta^{d-1})
		|0\rangle
\end{align}
	with $\varphi^1, \ldots,\varphi^{d-1}\in [0,2\pi)$ and $\theta^1,\ldots,\theta^{d-1} \in [0,\pi/2]$.
	%with $\varphi^{s} \in [0,2\pi)$ and $\theta^s \in [0,\pi/2]$ as 

\begin{prot}[TQOT   protocol for qudit pure states] \Label{PR4}
For any message pure states $|\psi_1\rangle,\ldots, |\psi_{\sff}\rangle \in \pureset{\mathbb{C}^d}$,
	we choose the parameters $\varphi_\ell^1,\ldots, \varphi_\ell^{d-1}$ and $\theta_\ell^1,\ldots, \theta_\ell^{d-1}$ as
	%where each state $|\psi_\ell\rangle$ is written as 
\begin{align}	
	|\psi_\ell\rangle = 
		\Ss(\varphi^1_{\ell},\ldots, \varphi^{d-1}_{\ell})
		\Rr(\theta^1_{\ell},\ldots, \theta^{d-1}_{\ell})
		|0\rangle.\Label{CAP}
\end{align}
When the user's target index $K$ is $k\in[\sff]$, i.e., the targeted state is $|\psi_k\rangle$, our protocol is given as follows.
%The user's target index is $k\in[\sff]$, i.e., the targeted state is $|\psi_k\rangle$.
%For simplicity, we denote $(\varphi^j,\theta^j) \coloneqq (\varphi^j_{k},\theta^j_{k})$ without the superscript
%	and $\bm{\varphi} \coloneqq (\varphi^1,\ldots,\varphi^{d-1})$.
\begin{description}[leftmargin=1.5em]
\item[0)] \textbf{Entanglement Sharing}: 
			The same as Entanglement Sharing of Protocol~\ref{PR2}.
\item[1)] \textbf{Query 1}: 
			The same as Query of Protocol~\ref{PR1}.
\item[2)] \textbf{Answer 1}: 
			The same as Answer of Protocol~\ref{PR2}.
\item[3)] \textbf{Reconstruction 1}:
	\begin{enumerate}
		%\item[a)]
		\item
The same as Reconstruction of Protocol~\ref{PR2}. We denote the output system ${\cal H}_{d}$.
		\item
When $q_k=1$, the user sets qudits $B,B'$ to be the completely mixed state.
Then, he applies the isometry $\sV$ from ${\cal H}_{d}$ to $A \otimes A'$.
\if0
, where
the isometry $\sV_3$ is defined as
\begin{align}
\sV_3 |j\rangle =|j\rangle|j\rangle 
\end{align}
\fi
When $q_k=0$, the user sets qudits $A,A'$ to be the completely mixed state.
Then, he applies the isometry $\sV$ from ${\cal H}_{d}$ to $B \otimes B'$.
\end{enumerate}
\item[4)] \textbf{Query 2}: 
The user sends the system $A\otimes B$ ($A'\otimes B'$) to Server 1 (Server 2).
\item[5)] \textbf{Answer 2}: 
				{When $Q=q$ and $Q'=q'$,} 
				Server $1$ applies 
				$\Ss( \sum_{\ell=1}^\sff q_\ell \varphi_\ell ^1,\ldots,
				 \sum_{\ell=1}^\sff q_\ell \varphi_\ell ^{d-1})\otimes 
				\Ss( -\sum_{\ell=1}^\sff q_\ell \varphi_\ell ^1,\ldots, 
				-\sum_{\ell=1}^\sff q_\ell \varphi_\ell ^{d-1})$
				on $A\otimes B$,
				and sends $A\otimes B$ to the user.
				Similarly, Server $2$ applies 
				$\Ss( -\sum_{\ell=1}^\sff q_\ell' \varphi_\ell^1,\ldots,
				-\sum_{\ell=1}^\sff q_\ell' \varphi_\ell^{d-1})\otimes 
				\Ss( \sum_{\ell=1}^\sff q_\ell' \varphi_\ell^1,\ldots,
				 \sum_{\ell=1}^\sff q_\ell' \varphi_\ell^{d-1})$
				on $A'\otimes B'$,
				and sends $A'\otimes B'$ to the user.
\item[6)] \textbf{Reconstruction 2}:
When $q_k=1$, the user discards $B,B'$.
Then, he applies the partial isometry $\sV^\dagger$ from $A \otimes A'$ to $A$.
When $q_k=0$, the user discards $A,A'$.
Then, he applies the partial isometry $\sV^\dagger $ from $B \otimes B'$ to $A$.
				\QEDA
\end{description}
\end{prot}

Protocol~\ref{PR4} satisfies the correctness, secrecy, and communication complexity, which is shown as follows.

\begin{lemm} \Label{LPR4}
Protocol \ref{PR4} is a correct TQOT  protocol.
Its upload complexity and its download complexity
are $2\sff$ bits plus $4 \log d$ qubits
and $2(d-1)+4 \log d$ qubits, respectively.
The required prior entanglement is $(d-1)$ copies of $|I_2\rrangle$
and two copies of $|I_d\rrangle$, i.e., 
$(d-1)+2\log d$ ebits.
\end{lemm}
\begin{proof}
To show correctness,
we assume that the servers and the user are honest.
When $q_k = 1$,
at the end of Step 3), the state on qudits $A \otimes A'$ is
$\sV  \big(\Rr(\theta_k^1, \ldots, \theta_k^{d-1}) |0\rangle\big)|0\rangle$.
At the end of protocol, the state on $A$ is
\begin{align}
&\sV^\dagger
\Ss \Bigg( \sum_{\ell=1}^\sff (q_\ell-q_\ell') \varphi_\ell ^1, \ldots, 
\sum_{\ell=1}^\sff (q_\ell-q_\ell') \varphi_\ell^{d-1} \Bigg) \nonumber\\
&\cdot \sV
\big(\Rr(\theta_k^1, \ldots, \theta_k^{d-1}) |0\rangle\big) |0\rangle \nonumber\\
=&\sV^\dagger
\Ss( \varphi_k^1, \ldots,  \varphi_k^{d-1} ) 
\sV
\big(\Rr(\theta_k^1, \ldots, \theta_k^{d-1}) |0\rangle\big)|0\rangle \nonumber\\
=&  \big(\Ss( \varphi_k^1, \ldots, \varphi_k^{d-1} ) 
\Rr(\theta_k^1, \ldots, \theta_k^{d-1}) |0\rangle\big) |0\rangle.
\end{align}
When $q_k = 0$,
at the end of Step 3), the state on $B \otimes B'$ is
$\sV \big( \Rr(\theta_k^1, \ldots, \theta_k^{d-1}) |0\rangle\big)|0\rangle$.
At the end of protocol, the state on $A$ is
\begin{align}
&\sV^\dagger
\Ss \Bigg( \sum_{\ell=1}^\sff (-q_\ell+q_\ell') \varphi_\ell ^1, \ldots, 
\sum_{\ell=1}^\sff (-q_\ell+q_\ell') \varphi_\ell^{d-1} \Bigg) \nonumber\\
&\cdot \sV
\big(\Rr(\theta_k^1, \ldots, \theta_k^{d-1}) |0\rangle \big) |0\rangle\nonumber\\
=&\sV^\dagger
\Ss( \varphi_k^1, \ldots,  \varphi_k^{d-1} ) 
\sV
\big( \Rr(\theta_k^1, \ldots, \theta_k^{d-1}) |0\rangle \big) |0\rangle\nonumber\\
=&  \big(\Ss( \varphi_k^1, \ldots, \varphi_k^{d-1} ) 
\Rr(\theta_k^1, \ldots, \theta_k^{d-1}) |0\rangle\big) |0\rangle.
\end{align}
\end{proof}

In fact, when the user and the servers are honest, 
the user has only the system $A$ at the end of the protocol
so that the user has no information for other $\psi_\ell$.
Also, since both servers do not make any measurement in Answer 2,
both severs obtain information only from Query 1.
Since Query is the same as Protocol \ref{PR1}, 
both servers have no information for the user's choice $K$.
However, when one of them is not honest,
the secrecy does not holds.

\section{Secrecy Problems in Protocol \ref{PR4}}\Label{S6-H}
This subsection presents several attacks by 
a malicious user or a malicious server in Protocol \ref{PR4}.
First, 
we show an attack by a malicious server in Protocol \ref{PR4}.
Under Protocol \ref{PR4},
we assume that the user and Server 2 are honest, but Server 1 is malicious.
Consider that Server 1 intends to identify whether $K$ is 1 or not.
Hence, we consider the following behavior of Server 1.
In Answer 1, Server $1$ applies 
$\Rr( (\sum_{\ell=1}^\sff q_\ell \theta_\ell^j )-\theta_1^j )$	on $A^j$
for $j=1, \ldots, d-1$.
In Answer 2, Server $1$ measures the received systems $A$ and $B$
with the basis 
$\{|0\rangle,\ldots,|d-1\rangle\}$, and sends back the resultant state.
When one of the outcomes does not correspond to the basis $|0\rangle$,
Server $1$ finds that  the variable $K$ is not $1$.
This is because the outcome $b$ is $0$ when the variable $K$ is $1$.
Hence, Protocol \ref{PR4} does not have the user secrecy.

Next, we present an attack by a malicious user in Protocol \ref{PR4}.
Under Protocol \ref{PR4},
we assume that  both servers are honest, but the user is malicious.
Consider that the user intends to get 
the state $\Ss(\varphi_{k'}^1, \ldots, \varphi_{k'}^{d-1}) 
\frac{1}{\sqrt{d}}\sum_{j=0}^{d-1}|j\rangle$
with $k'\neq k$ in addition to the state $|\psi_k\rangle$.
In Query 1, the user sets $q_{k'}$ to be $1$ and $q_j$ to be $0$ for $j\neq k'$.
In Reconstruction 1, the user makes the same operation as the honest case.
In Query 2, the user sets the systems $A$ and $A'$ to be the state 
$\frac{1}{\sqrt{d}}\sum_{j=0}^{d-1}|j\rangle$ instead of the completely mixed state.
Also, the user sets the system $B\otimes B'$ in the same way as the honest case.
In Reconstruction 2, the user makes the same operation to $B \otimes B'$
as the honest case, which outputs the state $|\psi_k\rangle$.
Also, the user keeps the system $A$, whose state is 
the state $\Ss(\varphi_{k'}^1, \ldots, \varphi_{k'}^{d-1}) 
\frac{1}{\sqrt{d}}\sum_{j=0}^{d-1}|j\rangle$.
Hence, Protocol \ref{PR4} does not have the server secrecy.

%\section{Two-Server Probabilistic Symmetric QPIR Protocols in Visible Setting} \Label{sec:visible}

\section{Probabilistic TQOT  protocol for pure qubit states}\Label{S7-C}
As shown in Section \ref{S6-H}, Protocol \ref{PR4} does not have 
the user secrecy nor the server secrecy.
In Protocol \ref{PR4}, the user can control the receiving state by modifying the state sent to the server.
To avoid this problem, we consider a protocol, in which, the user sends only the classical information to the servers.
That is, modifying Protocol \ref{PR4} slightly,
we construct a TQOT protocol for pure qudit states in the visible setting which achieve
the user secrecy and the server secrecy.
%the communication complexity in Theorem~\ref{theo:31}.
%Now, we construct a TQOT  protocol in the visible setting for qubit states, which achieves the communication complexity in Theorem~\ref{theo:31}.

\begin{prot}[TQOT  protocol for qudit pure states] \Label{PR5}
For any message pure states $|\psi_1\rangle,\ldots, |\psi_{\sff}\rangle \in \pureset{\mathbb{C}^d}$,
	we choose the parameters $\varphi_\ell^1,\ldots, \varphi_\ell^{d-1}$ and $\theta_\ell^1,\ldots, \theta_\ell^{d-1}$ as \eqref{CAP}.
When the user's target index $K$ is $k\in[\sff]$, i.e., the targeted state is $|\psi_k\rangle$, our protocol is given as follows.
%The user's target index is $k\in[\sff]$, i.e., the targeted state is $|\psi_k\rangle$.
%For simplicity, we denote $(\varphi^j,\theta^j) \coloneqq (\varphi^j_{k},\theta^j_{k})$ without the superscript
%	and $\bm{\varphi} \coloneqq (\varphi^1,\ldots,\varphi^{d-1})$.

\begin{description}[leftmargin=1.5em]
\item[0)] \textbf{Entanglement Sharing}: 
	Let $A_{j,1},\ldots, A_{j,d-1} , A_{j,1}',\ldots, A_{j,d-1}'$ and 
	$A_j,A_j',B_j,B_j'$ be qubits and qudits, respectively, for $j=1, \ldots,n$ with sufficiently large $n$.
	Before starting the protocol,
	Server 1 and Server 2 share 
				$n(d-1)$ maximally entangled state $|\sI_2\rrangle$ on 
				$A_{j,1}\otimes A_{j,1}', \ldots, A_{j,d-1}\otimes A_{j,d-1}'$ with $j=1,2,\ldots,n$,
%	Server 1 and Server 2 share 
%				two maximally entangled state $|I\rrangle$ on $A\otimes A'$ and $B\otimes B'$,
				where Server 1 (Server 2) contains $A_{j,1},\ldots, A_{j,d-1}$ ($A_{j,1}',\ldots, A_{j,d-1}'$).
	Server 1 and Server 2 share 
$2n$ copies of the maximally entangled state $|\sI_d\rrangle$ on 
				$A_j\otimes A_j', B_j\otimes B_j' $ for $j=1,2,\ldots,n$.
\item[1)] \textbf{Query 1}: 
			The same as Protocol~\ref{PR1}.
%			The user chooses 
%				$Q = (q_{1},\ldots, q_{\sff}) \in  \{0,1\}^{\sff}$ randomly 
%				and define $Q' = (q_{1}',\ldots,q_{\sff}' )\in  \{0,1\}^{\sff}$ as 
%				\begin{align}
%				q_{i}' = \begin{cases}  q_{i} & \text{for $i\neq k$},	\\ q_{i}\oplus 1 & \text{for $i=k$}. \end{cases}
%				\end{align}
%			The user sends $Q$ and $Q'$ to Server 1 and Server 2, respectively.

The following steps are given inductively for $j=1,2,\ldots$ up to stopping the protocol. 
\item[3j-1)] \textbf{Answer j}: 
	\begin{enumerate}
\item The servers perform the same operations on $A_{j,1},\ldots, A_{j,d-1} , A_{j,1}',\ldots, A_{j,d-1}'$ 
as Answer of Protocol~\ref{PR2}.
\item The servers perform the same operations on $A_j,A_j',B_j,B_j'$ as Answer 2 of Protocol~\ref{PR4}.
	\end{enumerate}
\item[3j)] \textbf{Reconstruction j}:
	\begin{enumerate}
		%\item[a)]
		\item
			The user applies the same operation to $A_{j,1},\ldots, A_{j,d-1} , A_{j,1}',\ldots, A_{j,d-1}'$ 
			as Reconstruction of Protocol~\ref{PR2}.
			Then, the user obtains the system ${\cal H}_{1,d}, \ldots, 
			{\cal H}_{n,d}$ as its outputs.
		%\item[a)]
		\item 
		When $q_k=1$, the user measures the system 
		${\cal H}_{j,d}\otimes  A_j'$ with the basis
$\mathbf{M}_{\sX\sZ,d} = \{ |\sX^a \sZ^b\rrangle \mid a,b \in [0:d-1] \}$,  
obtains the outcome $(a_j,b_j)$, and applies $\sZ^{-b_j}$.
When $a_j=0$, the user keeps the system $A_j$ as the final output system.

		When $q_k=0$, the user measures the system 
		${\cal H}_{j,d}\otimes  B_j'$ with the basis
$\mathbf{M}_{\sX\sZ,d}$, 
		obtains the outcome $(a_j,b_j)$, and  and applies $\sZ^{-b_j}$.
When $a_j=0$, the user keeps the system $B_j$ as the final output system.
\end{enumerate}
\item[3j+1)] \textbf{Query j+1}:
The user informs the servers whether $a_j$ is $0$ or not.
If $a_j$ is $0$, the protocol is terminated.
Otherwise, the protocol proceed to Step 3j+2).
\QEDA
\end{description}
\end{prot}

Protocol~\ref{PR5} satisfies the correctness, secrecy, and communication complexity, which is shown as follows.

\begin{lemm} \Label{LPR5}
Protocol \ref{PR5} is a correct TQOT  protocol that satisfies 
the server secrecy.
Its upload complexity and its download complexity are $2\sff+ 2d$ bits 
and $(2(d-1)+4 \log d)d$ qubits, respectively, in average.
The consumed prior entanglement is 
$d(d-1) $ copies of $|I_2\rrangle$ and 
$2d$ copies of $|I_d\rrangle$, i.e., 
$(d-1+2 \log d)d$ ebits in average.
\end{lemm}
\begin{proof}
The correctness is shown as follows.
We assume that the servers and the user are honest.
At the end of Answer, for $j=1,2, \ldots,n$,
the user receives 
the states 
$|\Rr((-1)^{q_k+1}\theta_k^1)\rrangle, \ldots, 
|\Rr((-1)^{q_k+1}\theta_k^{d-1})\rrangle$
on $A_{j,1}\otimes A_{j,1}',\ldots,  A_{j,d-1}\otimes A_{j,d-1}'$,
$|\Ss((-1)^{q_k+1}\varphi_k^{1}, \ldots, (-1)^{q_k+1}\varphi_k^{d-1})
\rrangle$ on $A_j \otimes A_j'$,
and
$|\Ss((-1)^{q_k}\varphi_k^{1}, \ldots, (-1)^{q_k}\varphi_k^{d-1})
\rrangle$ on $B_j \otimes B_j'$.
At the end of 1) of Reconstruction, the user has 
the states 
$\Rr(\varphi_k^{1}, \ldots, \varphi_k^{d-1})|0\rangle$ on ${\cal H}_{j,d}$,
$|\Ss((-1)^{q_k+1}\varphi_k^{1}, \ldots, (-1)^{q_k+1}\varphi_k^{d-1})
\rrangle$ on $A_j \otimes A_j'$,
and
$|\Ss((-1)^{q_k}\varphi_k^{1}, \ldots, (-1)^{q_k}\varphi_k^{d-1})
\rrangle$ on $B_j \otimes B_j'$.
When $q_k=1$ and $a_j=0$, 
at the end of 2) of Reconstruction, the user has 
$\sZ^{-b_j} \Ss(\varphi_k^{1}, \ldots, \varphi_k^{d-1})
\sZ^{b_j} \Rr(\varphi_k^{1}, \ldots, \varphi_k^{d-1})|0\rangle
=\Ss(\varphi_k^{1}, \ldots, \varphi_k^{d-1})
\Rr(\varphi_k^{1}, \ldots, \varphi_k^{d-1})|0\rangle$ on ${\cal H}_{j,d}$.
When $q_k=1$ and $a_j=0$, we have the same characterization.

Therefore, when the user obtains the outcome $a_j=0$,
the user recovers  $\Ss(\varphi_k^{1}, \ldots, \varphi_k^{d-1})
\Rr(\varphi_k^{1}, \ldots, \varphi_k^{d-1})|0\rangle$.

The complexity is calculated as follows.
This protocol has Query j+1 with the probability $(\frac{d-1}{d})^{j-1}$.
In average, the upload complexity is $2\sff+ 2\sum_{j=1}^{\infty}(\frac{d-1}{d})^{j-1}
=2\sff+ 2d$ bits.
This protocol has Answer j with the probability $(\frac{d-1}{d})^{j-1}$.
In average, the download complexity is $
(2(d-1)+4 \log d)\sum_{j=1}^{\infty}(\frac{d-1}{d})^{j-1}
=(2(d-1)+4 \log d)d $ qubits.
Hence, the consumed prior entanglement is 
$d(d-1) $ copies of $|I_2\rrangle$ and 
$2d$ copies of $|I_d\rrangle$, i.e., 
$(d-1+2 \log d)d$ ebits in average.

Assume that the servers are honest.
At the end of Answer, the user has the states
$|\Rr( \sum_{\ell=1}^\sff (q_\ell -q_\ell' )\theta_\ell^j)\rrangle$ for $j=1, \ldots, d-1$,
$|\Ss( \sum_{\ell=1}^\sff (q_\ell -q_\ell' ) \varphi_\ell^1, \ldots,  
\sum_{\ell=1}^\sff (q_\ell -q_\ell' ) \varphi_\ell^{d-1}) \rrangle$, and 
$|\Ss( -\sum_{\ell=1}^\sff (q_\ell -q_\ell' ) \varphi_\ell^1, \ldots,  
-\sum_{\ell=1}^\sff (q_\ell -q_\ell' ) \varphi_\ell^{d-1}) \rrangle$.
In order to recover the state 
$\Ss( \varphi_k^1, \ldots, \varphi_k^{d-1}) \Rr(\theta_k^1,\ldots, \theta_k^{d-1})|0\rangle$,
these states need to contains the information for $\varphi_k^1, \ldots, \varphi_k^{d-1}$ and 
$\theta_k^1, \ldots, \theta_k^{d-1}$.
This condition holds only when 
$\sum_{\ell=1}^\sff (q_\ell -q_\ell' )\theta_\ell^1,
\ldots, 
\sum_{\ell=1}^\sff (q_\ell -q_\ell' )\theta_\ell^{d-1},$
$\sum_{\ell=1}^\sff (q_\ell -q_\ell' ) \varphi_\ell^1,\ldots,
\sum_{\ell=1}^\sff (q_\ell -q_\ell' ) \varphi_\ell^{d-1}$
are constant times of $\varphi_k^1,\ldots,\varphi_k^{d-1}$ and 
$\theta_k^1,\ldots,\theta_k^{d-1}$, respectively.
This condition shows the state in the user's hand does not depend on 
$\varphi_\ell^1,\ldots,\varphi_\ell^{d-1}$, $\theta_\ell^1,\ldots, \theta_\ell^{d-1}$ 
for $\ell \neq k$.
Hence, the server secrecy holds.
\end{proof}

Indeed, when both servers are honest, 
each server does not obtain any information of $k$ as follows.
In this case,
the outcome $a_j$ is subject to the uniform distribution.
Both servers cannot obtain any information from Query 2.
Query 1 is uniformly random $\sff$ bits.
Hence, each server does not obtain any information of $k$ as follows.

However,
Protocol \ref{PR5} does not have the user secrecy as follows.
%under the malicious-server model 
Under Protocol \ref{PR5},
we assume that the user and Server 2 are honest, but Server 1 is malicious.
Consider that Server 1 intends to identify whether $K$ is 1 or not.
Hence, we consider the following behavior of Server 1.
In Answer 1, Server $1$ applies $\Rr( (\sum_{\ell=1}^\sff q_\ell \theta_\ell^j )-\theta_1^j )$ on $A_{j,1}$.
In addition, Server $1$ replaces the state on $B_1$ by $|0\rangle$.
If the outcome $a_1$ in Query 1 is not $0$, 
Server $1$ finds that the variable $K$ is not $1$.
This is because the outcome $a_j$ is $0$ when the variable $K$ is $1$.
Hence, Protocol \ref{PR5} does not have the user secrecy.
This problem needs to be resolved.

\section{Probabilistic TQOT  protocol for pure qudit states}
Protocols \ref{PR4} and \ref{PR5} do not achieve
the user secrecy.
If the query is composed only of the same query as Query 1 of Protocol \ref{PR1},
the user secrecy holds.
In the following, 
as a protocol to satisfy the above condition,
we construct a probabilistic TQOT protocol for 
pure qudit states in the visible setting that
achieves 
the user secrecy and the server secrecy.
The following protocol achieves the required properties.
%the communication complexity in Theorem~\ref{theo:31}.

\begin{prot}[TQOT   protocol for qudit pure states] \Label{PR6}
For any message pure states $|\psi_1\rangle,\ldots, |\psi_{\sff}\rangle \in \pureset{\mathbb{C}^d}$,
	we choose the parameters $\varphi_\ell^1,\ldots, \varphi_\ell^{d-1}$ and $\theta_\ell^1,\ldots, \theta_\ell^{d-1}$ as \eqref{CAP}.
When the user's target index $K$ is $k\in[\sff]$, i.e., the targeted state is $|\psi_k\rangle$, our protocol is given as follows.
%The user's target index is $k\in[\sff]$, i.e., the targeted state is $|\psi_k\rangle$.
%For simplicity, we denote $(\varphi^j,\theta^j) \coloneqq (\varphi^j_{k},\theta^j_{k})$ without the superscript
%	and $\bm{\varphi} \coloneqq (\varphi^1,\ldots,\varphi^{d-1})$.

\begin{description}[leftmargin=1.5em]
\item[0)] \textbf{Entanglement Sharing}: 
	Let $A_{j,1},\ldots, A_{j,d-1} , A_{j,1}',\ldots, A_{j,d-1}'$ and 
	$A_j,A_j',B_j,B_j'$ be qubits and qudits, respectively, for $j=1, \ldots,n$.
	Before starting the protocol,
	Server 1 and Server 2 share 
				$n(d-1)$ maximally entangled state $|\sI_2\rrangle$ on 
				$A_{j,1}\otimes A_{j,1}', \ldots, A_{j,d-1}\otimes A_{j,d-1}'$ with $j=1,2,\ldots,n$,
%	Server 1 and Server 2 share 
%				two maximally entangled state $|I\rrangle$ on $A\otimes A'$ and $B\otimes B'$,
				where Server 1 (Server 2) contains $A_{j,1},\ldots, A_{j,d-1}$ ($A_{j,1}',\ldots, A_{j,d-1}'$).
	Server 1 and Server 2 share 
$2n$ copies of the maximally entangled state $|\sI_d\rrangle$ on 
				$A_j\otimes A_j', B_j\otimes B_j' $ for $j=1,2,\ldots,n$.
\item[1)] \textbf{Query}: 
			The same as Protocol~\ref{PR1}.
%			The user chooses 
%				$Q = (q_{1},\ldots, q_{\sff}) \in  \{0,1\}^{\sff}$ randomly 
%				and define $Q' = (q_{1}',\ldots,q_{\sff}' )\in  \{0,1\}^{\sff}$ as 
%				\begin{align}
%				q_{i}' = \begin{cases}  q_{i} & \text{for $i\neq k$},	\\ q_{i}\oplus 1 & \text{for $i=k$}. \end{cases}
%				\end{align}
%			The user sends $Q$ and $Q'$ to Server 1 and Server 2, respectively.
\item[2)] \textbf{Answer}:
The servers make the same operation as Answer j of Protocol \ref{PR5} for $j=1,2,\ldots, n$.  
\if0
	\begin{enumerate}
\item The servers perform the same operations on $A_{j,1},\ldots, A_{j,d-1} , A_{j,1}',\ldots, A_{j,d-1}'$ 
as Answer of Protocol~\ref{PR2} for $j=1,2,\ldots, n$.
\item The servers perform the same operations on $A_j,A_j',B_j,B_j'$ as Answer 2 of Protocol~\ref{PR4} for $j=1,2,\ldots, n$.
	\end{enumerate}
\fi
\item[3)] \textbf{Reconstruction}:
	\begin{enumerate}
		%\item[a)]
		\item
			The user applies the same operation to $A_{j,1},\ldots, A_{j,d-1} , A_{j,1}',\ldots, A_{j,d-1}'$ 
			as Reconstruction of Protocol~\ref{PR2} for $j=1,2, \ldots, n$.
			Then, the user obtains the system ${\cal H}_{1,d}, \ldots, 
			{\cal H}_{n,d}$ as its outputs.
		%\item[a)]
		\item 
		When $q_k=1$, the user measures the system 
		${\cal H}_{j,d}\otimes  A_j'$ with the basis
$\mathbf{M}_{\sX\sZ,d} = \{ |\sX^a \sZ^b\rrangle \mid a,b \in [0:d-1] \}$,  
obtains the outcome $(a_j,b_j)$, and applies $\sZ^{-b_j}$ for $j=1,2, \ldots, n$.
When $a_j=0$, the user keeps the system $A_j$ as the final output system.

		When $q_k=0$, the user measures the system 
		${\cal H}_{j,d}\otimes  B_j'$ with the basis
$\mathbf{M}_{\sX\sZ,d}$, 
		obtains the outcome $(a_j,b_h)$, and  and applies $\sZ^{-a_j}$ for $j=1,2, \ldots, n$.
When $b_j=0$, the user keeps the system $B_j$ as the final output system.
\QEDA
\end{enumerate}
\end{description}
\end{prot}

%The correctness, secrecy, and communication complexity of Protocol~\ref{PR4} are analyzed as follows.

Protocol~\ref{PR6} satisfies the correctness, secrecy, and communication complexity, which is shown as follows.

\begin{lemm} \Label{LPR6}
Protocol \ref{PR6} is a $1-(\frac{d-1}{d})^{n}$-correct TQOT  protocol that satisfies 
the user secrecy and the server secrecy. 
Its upload complexity and its download complexity are $2\sff$ bits 
and $2n(d-1)+4n\log d$ qubits, respectively.
The required prior entanglement is 
$n(d-1)$ copies of $|I_2\rrangle$ and 
$2n$ copies of $|I_d\rrangle$, i.e., 
$n(d-1)+2n\log d$ ebits.
\end{lemm}
\begin{proof}
The correctness can be shown as the same way as the correctness of Protocol \ref{PR5} in Lemma \ref{LPR5}.
\if0
is shown as follows.
We assume that the servers and the user are honest.
At the end of Answer, for $j=1,2, \ldots,n$,
the user receives 
the states 
$|\Rr((-1)^{q_k+1}\theta_k^1)\rrangle, \ldots, 
|\Rr((-1)^{q_k+1}\theta_k^{d-1})\rrangle$
on $A_{j,1}\otimes A_{j,1}',\ldots,  A_{j,d-1}\otimes A_{j,d-1}'$,
$|\Ss((-1)^{q_k+1}\varphi_k^{1}, \ldots, (-1)^{q_k+1}\varphi_k^{d-1})
\rrangle$ on $A_j \otimes A_j'$,
and
$|\Ss((-1)^{q_k}\varphi_k^{1}, \ldots, (-1)^{q_k}\varphi_k^{d-1})
\rrangle$ on $B_j \otimes B_j'$.
At the end of 1) of Reconstruction, the user has 
the states 
$\Rr(\varphi_k^{1}, \ldots, \varphi_k^{d-1})|0\rangle$ on ${\cal H}_{j,d}$,
$|\Ss((-1)^{q_k+1}\varphi_k^{1}, \ldots, (-1)^{q_k+1}\varphi_k^{d-1})
\rrangle$ on $A_j \otimes A_j'$,
and
$|\Ss((-1)^{q_k}\varphi_k^{1}, \ldots, (-1)^{q_k}\varphi_k^{d-1})
\rrangle$ on $B_j \otimes B_j'$.
When $q_k=1$ and $b_j=0$, 
at the end of 2) of Reconstruction, the user has 
$\sZ^{-a_j} \Ss(\varphi_k^{1}, \ldots, \varphi_k^{d-1})
\sZ^{a_j} \Rr(\varphi_k^{1}, \ldots, \varphi_k^{d-1})|0\rangle
=\Ss(\varphi_k^{1}, \ldots, \varphi_k^{d-1})
\Rr(\varphi_k^{1}, \ldots, \varphi_k^{d-1})|0\rangle$ on ${\cal H}_{j,d}$.
When $q_k=1$ and $b_j=0$, we have the same characterization.
\fi
That is, when the user obtains the outcome $b_j=0$ at least with one element $j$ among $1,\ldots,n$,
the user recovers  $\Ss(\varphi_k^{1}, \ldots, \varphi_k^{d-1})
\Rr(\varphi_k^{1}, \ldots, \varphi_k^{d-1})|0\rangle$.
Otherwise, the user cannot recover it.
This protocol works correctly with probability $1-(\frac{d-1}{d})^{n}$.

The complexity is calculated as follows.
The upload complexity is $2\sff$ bits.
The download complexity is $2n(d-1)+4n \log d$ qubits.

The secrecy is shown as follows.
Since each server receives uniformly random $\sff$ bits as Query,
the server does not obtain any information of $k$.
Hence, the user secrecy holds.

Assume that the servers are honest.
At the end of Answer, the user has the states
$|\Rr( \sum_{\ell=1}^\sff (q_\ell -q_\ell' )\theta_\ell^j)\rrangle$ for $j=1, \ldots, d-1$,
$|\Ss( \sum_{\ell=1}^\sff (q_\ell -q_\ell' ) \varphi_\ell^1, \ldots,  
\sum_{\ell=1}^\sff (q_\ell -q_\ell' ) \varphi_\ell^{d-1}) \rrangle$, and 
$|\Ss( -\sum_{\ell=1}^\sff (q_\ell -q_\ell' ) \varphi_\ell^1, \ldots,  
-\sum_{\ell=1}^\sff (q_\ell -q_\ell' ) \varphi_\ell^{d-1}) \rrangle$.
In order to recover the state 
$\Ss( \varphi_k^1, \ldots, \varphi_k^{d-1}) \Rr(\theta_k^1,\ldots, \theta_k^{d-1})|0\rangle$,
these states need to contains the information for $\varphi_k^1, \ldots, \varphi_k^{d-1}$ and 
$\theta_k^1, \ldots, \theta_k^{d-1}$.
This condition holds only when 
$\sum_{\ell=1}^\sff (q_\ell -q_\ell' )\theta_\ell^1,
\ldots, 
\sum_{\ell=1}^\sff (q_\ell -q_\ell' )\theta_\ell^{d-1},$
$\sum_{\ell=1}^\sff (q_\ell -q_\ell' ) \varphi_\ell^1,\ldots,
\sum_{\ell=1}^\sff (q_\ell -q_\ell' ) \varphi_\ell^{d-1}$
are constant times of $\varphi_k^1,\ldots,\varphi_k^{d-1}$ and 
$\theta_k^1,\ldots,\theta_k^{d-1}$, respectively.
This condition shows the state in the user's hand does not depend on 
$\varphi_\ell^1,\ldots,\varphi_\ell^{d-1}$, $\theta_\ell^1,\ldots, \theta_\ell^{d-1}$ 
for $\ell \neq k$.
Hence, the server secrecy holds.
\end{proof}

\section{Two-Server Symmetric QPIR Protocols with Mixed States in Visible Setting} \Label{sec:mix}
The TQOT  protocols in the previous sections are for the retrieval of pure states.
The aim of this section is showing the following theorem, which works with mixed states and has scalability.

\begin{theo} \Label{theo:mix}
For a positive number $0<\alpha<1$,
a large dimension $d$, and a large positive integer $\sff$, 
there exists a probabilistic $\alpha$-correct TQOT  protocol for mixed states with the following properties on $\mathbb{C}^d$.
%It is a $\alpha_d$-correct TQOT  protocol.
The upload and download complexity are $2\sff$ and $O(d ^2)$. 
It satisfies the user secrecy and the server secrecy. 
It needs $O(d^2)$-ebit prior entanglement.
\end{theo}

To show this theorem, we convert these protocols to TQOT  protocols for mixed states.
For this conversion, we first give a decomposition of mixed states, and then construct the protocol for mixed states.

    \subsection{Decomposition of mixed states}
If a protocol is based on the blind setting, it works with mixed states.
However, since our protocols in previous sections are based on the visible setting,
they do not work with mixed states
because they need the description of a pure state as the input.
To resolve this problem, 
we can consider the following method:
The servers randomly choose the pure state to be sent.
To accomplish this method,
we decompose a mixed state $\rho$ on a $d$-dimensional Hilbert space 
as $\rho = \sum_{i=0}^{d-1} p_i |\psi_i\rangle \langle \psi_i|$.
One canonical decomposition is given by using the diagonalization of 
$\rho$. 
To implement the above mentioned method based on this decomposition,
the servers have to share a random variable that subject to the 
distribution $\{p_i\}_i$.
However, in this method, the probabilities $p_i$ depend on the state $\rho$,
and take continuous values.
This idea does not work with the diagonalization of $\rho$. 
However, if the probability distribution $\{p_i\}_i$ is limited to the uniform distribution,
the above idea works well.

We choose the decomposition $(p_i,|\psi_i\rangle)_{i=0}^{d-1}$ for $\rho$
to satisfy $\rho = \sum_{i=0}^{d-1} p_i |\psi_i\rangle \langle \psi_i|$.
Generally, a state $\rho$ has various decomposition.
Based on the computation basis $\{|j\rangle\}_{j=0}^{d-1}$,
we uniquely choose the decomposition
according to the method given in Appendix~\ref{append:decomp}.
Then, we define the vector
    \begin{align}
    |\phi_j\rangle \coloneqq  \sum_{i=0}^{d-1} \omega^{ij} \sqrt{p_i} |\psi_i\rangle \quad (\forall j\in[0:d-1]),
    \Label{def:phijj}
    \end{align}
    where $\omega = \exp(2\pi\ii/ d)$ and $\iota = \sqrt{-1}$.
Then, the state $\rho$ is decomposed by the vectors in \eqref{def:phijj} as 
    \begin{align}
    \rho = \sum_{j=0}^{d-1} \frac{1}{d} |\phi_j\rangle\langle \phi_j|.
    \Label{def:decompositionphi}
    \end{align}
Notice that 
    this decomposition is unique
    because the vectors  $|\psi_0\rangle, \ldots, |\psi_{d-1}\rangle$ are 
    uniquely chosen from the state $\rho$.

\subsection{TQOT  protocol for mixed states}
Next, we construct two-server TQOT  protocols for mixed states $\rho_1,\ldots, \rho_{\sff}$ on $d$-dimensional Hilbert spaces
by converting Protocol \ref{PR6}.
The same conversion can be applied to Protocols \ref{PR4} and \ref{PR5}.
 
Without losing generality, we assume that the servers share the decomposition \eqref{def:decompositionphi} of the mixed states, e.g., by the process in Appendix~\ref{append:decomp}.
%Suppose that the two servers share the decomposition of 
%That is, before the protocol starts,
%    the servers share 
%    the decomposition \eqref{def:decompositionphi} of the states $\rho_y$ with $y\in[\sff]$, 
For the states $\rho_y$ with $y\in[\sff]$, 
    we denote the vectors in \eqref{def:phijj} as $|\phi_{y,0}\rangle, \ldots, |\phi_{y,d-1} \rangle$.
%We assume that 
%    and denote $T(\rho_y) = (|\phi_{y,0}\rangle, \ldots, |\phi_{y,d-1} \rangle )$ for any $y\in[\sff]$.

We define a TQOT  protocol for mixed states from Protocol~\ref{PR6} as follows.

\begin{prot}[TQOT   protocol for qudit mixed states] \Label{PR7}
For any message mixed states $\rho_1,\ldots, \rho_{\sff} $,
we choose pure states $|\phi_{1,j}\rangle,\ldots, |\phi_{\sff,j}\rangle \in \pureset{\mathbb{C}^d}$ for $j=0, \ldots, d-1$
according to the method given in Appendix~\ref{append:decomp}.
Then, for pure states $|\phi_{1,j}\rangle,\ldots, |\phi_{\sff,j}\rangle \in \pureset{\mathbb{C}^d}$,
we choose the parameters 
$\varphi_\ell^{1,j},\ldots, \varphi_\ell^{d-1,j}$ and 
$\theta_\ell^{1,j},\ldots, \theta_\ell^{d-1,j}$ as \eqref{CAP}
for $j=0, \ldots, d-1$.
When the user's target index $K$ is $k\in[\sff]$, i.e., the targeted state is 
$\rho_k$, our protocol is given as follows.
%The user's target index is $k\in[\sff]$, i.e., the targeted state is $|\psi_k\rangle$.
%For simplicity, we denote $(\varphi^j,\theta^j) \coloneqq (\varphi^j_{k},\theta^j_{k})$ without the superscript
%	and $\bm{\varphi} \coloneqq (\varphi^1,\ldots,\varphi^{d-1})$.

\begin{description}[leftmargin=1.5em]
\item[0)] \textbf{Entanglement Sharing}: 
	Let $A_{j,1},\ldots, A_{j,d-1} , A_{j,1}',\ldots, A_{j,d-1}'$ and 
	$A,A',A_j,A_j',B_j,B_j'$ be qubits and qudits, respectively, for $j=1, \ldots,n$.
	Before starting the protocol,
	Server 1 and Server 2 share 
				$n(d-1)$ maximally entangled state $|\sI_2\rrangle$ on 
				$A_{j,1}\otimes A_{j,1}', \ldots, A_{j,d-1}\otimes A_{j,d-1}'$ with $j=1,2,\ldots,n$,
%	Server 1 and Server 2 share 
%				two maximally entangled state $|I\rrangle$ on $A\otimes A'$ and $B\otimes B'$,
				where Server 1 (Server 2) contains $A_{j,1},\ldots, A_{j,d-1}$ ($A_{j,1}',\ldots, A_{j,d-1}'$).
	Server 1 and Server 2 share 
$2n+1$ copies of the maximally entangled state $|\sI_d\rrangle$ on 
				$A \otimes A'$
				and $A_j\otimes A_j', B_j\otimes B_j' $ for $j=1,2,\ldots,n$.
\item[1)] \textbf{Query}: 
			The same as Protocol~\ref{PR1}.
%			The user chooses 
%				$Q = (q_{1},\ldots, q_{\sff}) \in  \{0,1\}^{\sff}$ randomly 
%				and define $Q' = (q_{1}',\ldots,q_{\sff}' )\in  \{0,1\}^{\sff}$ as 
%				\begin{align}
%				q_{i}' = \begin{cases}  q_{i} & \text{for $i\neq k$},	\\ q_{i}\oplus 1 & \text{for $i=k$}. \end{cases}
%				\end{align}
%			The user sends $Q$ and $Q'$ to Server 1 and Server 2, respectively.
\item[2)] \textbf{Answer}: 
%	\begin{enumerate}
%\item[0)] 
Servers 1 and 2 measure the system $A$ and $A'$ with the computation basis $\{|j\rangle\}_{j=0}^{d-1}$, respectively
and obtain the common outcome $j$.
Then, Servers 1 and 2 makes the same as Answer of Protocol \ref{PR6}
with $|\psi_1\rangle= |\phi_{1,j}\rangle, \ldots, |\psi_{\sff}\rangle= 
|\phi_{\sff,j}\rangle$.
% 	\end{enumerate}
\item[3)] \textbf{Reconstruction}: 
The user makes the same operation as Reconstruction of Protocol \ref{PR6}.
%with $|\psi_1\rangle= |\phi_{1,j}\rangle, \ldots, |\psi_{\sff}\rangle= |\phi_{\sff,j}\rangle$.
\QEDA
\end{description}
\end{prot}

Entanglement Sharing step of Protocol \ref{PR7}
requires one more copy of the maximally entangled state $|\sI_d\rrangle$
in comparison with 
Entanglement Sharing step of Protocol \ref{PR6}.
Since the outcome $j$ in Answer step obeys the uniform distribution,
the relation \eqref{def:decompositionphi} and 
the correctness of Protocol \ref{PR6} guarantee
the correctness of Protocol \ref{PR7}.
Since the user's behavior of Protocol \ref{PR7}
is the same as that of Protocol \ref{PR6},
the user secrecy of Protocol \ref{PR6}
implies the user secrecy of Protocol \ref{PR7}.
The server secrecy of Protocol \ref{PR6} 
implies that the user cannot obtain any information for 
$ |\phi_{j',j}\rangle$ for $j'\neq k$.
Hence, the server secrecy of Protocol \ref{PR7} holds.
Also, Protocol \ref{PR7} has the same upload and download complexity
as those of Protocol \ref{PR6}.
In summary, we have the following lemma.

\begin{lemm} \Label{LPR7}
Protocol \ref{PR7} is a $1-(\frac{d-1}{d})^{n}$-correct TQOT  protocol that satisfies 
the user secrecy 
and the server secrecy. 
Its upload complexity and its download complexity
are $2\sff$ bits 
and $2n(d-1)+4n\log d$ qubits, respectively.
The required prior entanglement is 
$n(d-1)$ copies of $|I_2\rrangle$ and 
$2n$ copies of $|I_d\rrangle$, i.e., 
$n(d-1)+2n\log d$ ebits.
\end{lemm}

Theorem \ref{theo:mix} can be shown by applying 
Lemma \ref{LPR7} to the case with $n= -d \log(1-\alpha) $
because $1-(\frac{d-1}{d})^{n}$ converges to $\alpha$ under this choice.

\section{Conclusion} \Label{sec:conclusion}
We have constructed a TQOT protocol that enables the user to download 
the intended mixed state among $\sff$ mixed states.
Existing protocols work with only classical messages.
The proposed protocol satisfies 
the user secrecy and the server secrecy.
To construct this protocol, we have constructed several protocols that work only in submodels.

There are many open problems related to the study of QPIR for quantum messages.
The communication complexity of our protocols increases exponentially 
the number of qubits to be transmitted.
Thus, constructing more efficient TQOT  protocols for qudits is also an open problem.
Since our protocol has only two servers,
there is a possibility that 
the communication complexity can be decreased by the extension to more than two servers.
Studying this direction is an interesting future problem.
Interesting applications of our TQOT  protocols can also be considered for other communication and computation problems.
We leave these questions as another future problem.

Indeed, the papers \cite{SH19-2,SH20,AHPH20,ASHPHH21,SJ18,FHGHK17}
discussed the case with colluding servers.
Another interesting future problem is 
to extend our results to the case with colluding servers.

\section*{Acknowledgement}
MH was supported in part by the National
Natural Science Foundation of China (Grants No. 62171212) and
Guangdong Provincial Key Laboratory (Grant No. 2019B121203002).
SS was supported by JSPS Grant-in-Aid for JSPS Fellows No.\ JP20J11484.

% \appendices
% \section{Appendix}
% 
% 
% 
% 
% 
% 
% \begin{defi}
% We call $A$ is purified in $AB$ 
% 	if for any $C$ outside $AB$,  the state on $AC$ is a product state.
% \end{defi}
% 
% \begin{lemm}
% If $A$ is purified in $ABC$ and $ABD$,
% 	then $A$ is purified in $AB$.
% \end{lemm}
% 
% \begin{lemm}
% If $A$ and $B$ are purified in $ABC$,
% 	$AB$ is purified in $ABC$.
% \end{lemm}	
% 
% \begin{lemm}
% If $A$ is purified in $AB$,
% 	$A$ is still purified in $AB$ even after the local unitary operation on $A$.
% \end{lemm}

\appendix

\section{Diagonalization Process of Mixed States} \Label{append:decomp}

In \eqref{def:decompositionphi}, we introduced a decomposition of mixed states $\rho$
    but it depends on the choice and order of the orthogonal unit eigenvectors of $\rho$. 
In this appendix, we give one method to 
    uniquely determine the eigenvectors $|\psi_0\rangle, \ldots, |\psi_{d-1}\rangle$ and their order.
%Then, the decomposition \eqref{def:decompositionphi} of $\rho$ is uniquely determined.

We fix an orthonormal basis $\cB = \{ |0\rangle, \ldots, |d-1\rangle \}$
    and represent vectors $|\gamma\rangle$ with coordinates $|\gamma\rangle = (c_0,\ldots, c_{d-1})$ with respect to $\cB$.
Consider the spectral decomposition $\rho = \sum_{i=1}^t  q_{i} P_{i}$ with $q_1 < q_2 < \cdots < q_t$, where $P_{i}$ are orthogonal projections to eigenspaces $E_{i}$.
% Without losing generality, we assume that $q_1 < q_2 < \cdots < q_t$, i.e., the eigenspaces are indexed with the strictly increasing order.

%If $\dim E_{i} = 1$ for all $i$, 
%    we choose the unit eigenvector $|\gamma_{i}\rangle = (c_{i,0},\ldots, c_{i,d-1}) \in E_i$ such that 
%        the first nonzero coordinate $c_{i,j}$ is a positive real number.       
%Next, we reorder the eigenvectors as
%    $|\psi_{i}\rangle = |\gamma_{\pi(i)}\rangle$ with a permutation $\pi$ on $[0:d-1]$
%    %$|\gamma_{0}\rangle, \ldots, |\gamma_{d-1}\rangle$
%    %as
%    %$|\psi_{0}\rangle, \ldots, |\psi_{d-1}\rangle$
%    so that the corresponding eigenvalues $p_i \coloneqq q_{\pi(i)}$ satisfy $p_0 < p_1 < \cdots < p_{d-1}$.

If $\dim E_{i} = 1$ for all $i$, 
    we choose the unit eigenvector $|\psi_{i}\rangle = (c_{i,0},\ldots, c_{i,d-1}) \in E_i$ such that 
        the first nonzero coordinate $c_{i,j}$ is a positive real number.       
Then, the vectors $|\psi_{i}\rangle$ are uniquely determined and ordered.
%Next, we reorder the eigenvectors as
%    $|\psi_{i}\rangle = |\gamma_{\pi(i)}\rangle$ with a permutation $\pi$ on $[0:d-1]$
%    %$|\gamma_{0}\rangle, \ldots, |\gamma_{d-1}\rangle$
%    %as
%    %$|\psi_{0}\rangle, \ldots, |\psi_{d-1}\rangle$
%    so that the corresponding eigenvalues $p_i \coloneqq q_{\pi(i)}$ satisfy $p_0 < p_1 < \cdots < p_{d-1}$.

If there exist $\dim E_{i} \ge 2$,
    for all $i$ with $\dim E_{i} \ge 2$,
        we choose the orthonormal eigenvectors 
%         $|\gamma_{i,1}\rangle, \ldots, |\gamma_{i,\dim E_i}\rangle \in E_i$ as follows.
    $|\gamma_{i,1}\rangle = (c_{i,j,0},\ldots, c_{i,j,d-1}) \in E_i$ ($\forall j \in[\dim E_i]$) so that 
    the first nonzero coordinate  $c_{i,j,l_j}$ is positive real number
    and 
    $l_j<l_{j+1}$.
%     and $c_{i,j',l_j} = \delta_{j,j'}$.
Next, 
   % we reorder the vectors 
   % $\{  | \gamma_{i,j}\rangle  \}$ as $|\psi_{0}\rangle, \ldots, |\psi_{d-1}\rangle$
   % so that
   % the corresponding eigenvalues $p_i$ satisfy $p_0 \le p_1 \le \cdots \le p_{d-1}$
   %     and 
   we concatenate the vectors as 
 \begin{align*}
& (|\psi_1\rangle,\ldots, |\psi_{d-1}\rangle) \\
 \coloneqq &
    (|\gamma_{1,1}\rangle, \ldots, |\gamma_{1,\dim E_1}\rangle , 
        \ldots ,
     |\gamma_{t,1}\rangle, \ldots, |\gamma_{t,\dim E_t}\rangle ).
     \end{align*}
Then, the vectors $|\psi_{i}\rangle$ are uniquely determined and ordered.


\begin{thebibliography}{99}

% \bibitem{SH21}
% S. Song and M. Hayashi, 
% 	``Quantum Private Information Retrieval for Quantum Messages,'' 
% {\em Proc. 2021 IEEE Int. Symp. Information Theory (ISIT)}, 
% Melbourne, Victoria, Australia, 12–20 July 2021.  pp. 1052 -- 1057.


\bibitem{Kimble}
H. J. Kimble, 
``The quantum internet,''
{\em Nature}. vol. 453, 1023 -- 1030, (2008). 

%%%
\bibitem{Hayashi2007}
M.~Hayashi, K.~Iwama, H.~Nishimura, R.~Raymond, and S.~Yamashita, ``{Quantum Network Coding},'' in {\em STACS 2007 SE - 52} (W.~Thomas and P.~Weil, eds.),
vol.~4393 of {\em Lecture Notes in Computer Science}, pp.~610--621, Springer Berlin Heidelberg, 2007.

\bibitem{PhysRevA.76.040301}
M.~Hayashi, 
``{Prior entanglement between senders enables perfect quantum network coding with modification},'' 
{\em Phys. Rev. A}, vol.~76, no.~4,~40301, 2007.

\bibitem{Kobayashi2009}
H.~Kobayashi, F.~{Le Gall}, H.~Nishimura, and M.~R\"{o}tteler, ``{General
  Scheme for Perfect Quantum Network Coding with Free Classical
  Communication},'' in {\em Automata, Languages and Programming SE - 52}
  (S.~Albers, A.~Marchetti-Spaccamela, Y.~Matias, S.~Nikoletseas, and
  W.~Thomas, eds.), vol.~5555 of {\em Lecture Notes in Computer Science},
  pp.~622--633, Springer Berlin Heidelberg, 2009.

\bibitem{Leung2010}
D.~Leung, J.~Oppenheim, and A.~Winter, 
``{Quantum Network Communication; The Butterfly and Beyond},'' 
\emph{IEEE Trans.\ Inform.\ Theory}, 
vol.~56, no.~7,~3478--3490, 2010.

\bibitem{JFM11}
A. Jain, M. Franceschetti, and D. A. Meyer. ``On quantum network coding,'' 
{\em J. Math. Phys.}, vol. 52, 032201, 2011

\bibitem{SH18-2}
S. Song and M. Hayashi, ``Secure Quantum Network Code without Classical Communication,'' 
\emph{IEEE Trans.\ Inform.\ Theory}
vol. 66, no. 2, pp. 1178 -- 1192 (2020).

\bibitem{PhysRevLett.101.060401}
G.~Chiribella, G.~M. D'Ariano, and P.~Perinotti, 
``Quantum circuit architecture,'' 
{\em Phys. Rev. Lett.}, vol.~101,~060401, 2008.

\bibitem{PhysRevA.80.022339}
G.~Chiribella, G.~M. D'Ariano, and P.~Perinotti, 
``Theoretical framework for quantum networks,'' 
{\em Phys. Rev. A}, vol.~80,~022339, 2009.

\bibitem{HS20}
M. Hayashi and S. Song,
``Quantum state transmission over partially corrupted quantum information network,''
{\em Phys. Rev. Research} {\bf 2}, 033079 (2020).

\bibitem{BH20}
  N. Beaudrap and S. Herbert, ``Quantum linear network coding for entanglement distribution in restricted architectures'', {\em Quantum}, {Verein zur F{\"{o}}rderung des Open Access Publizierens in den Quantenwissenschaften}, vol 4, 356, 2020.

\bibitem{LIY19}
H. Lu, Z.-D. Li, X.-F. Yin, R. Zhang, X.-X. Fang, L. Li, N.-L. Liu, F. Xu, Y.-A. Chen, and J.-W. Pan,
``Experimental quantum network coding,''
{\em npj Quantum Inf} {\bf 5}, 89 (2019). 
 
\bibitem{PCXLY21}
	X. Pan,  X. Chen, G. Xu, Z. Li and Yixian Yang,
	``High dimensional quantum network coding based on prediction mechanism over the butterfly network'',
{\em Quantum Science and Technology}, {IOP} Publishing, vol. 7, 1, 015006, 2021.

\bibitem{NBA17}
H. V. Nguyen, Z. Babar, D. Alanis, P. Bosinis, D. Chandra, M. A. M. Izhar, S. X. Ng, and L. Hanzo,
``Towards the Quantum Internet: Generalised Quantum Network Coding for Large-Scale Quantum Communication Networks,'' 
{\em IEEE Access}, vol. 5, pp. 17288-17308, 2017.

\bibitem{PMS20}
P. Pathumsoot, T. Matsuo, T. Satoh, M. Hajdu\v{s}ek, S. Suwanna, and R. V. Meter,
``Modeling of measurement-based quantum network coding on a superconducting quantum processor,''
{\em Phys. Rev. A} {\bf 101}, 052301 (2020).

\bibitem{PXP21}
X.-B. Pan, G. Xu, Z.-P. Li, X.-B. Chen, and Y.-X. Yang,
``Quantum network coding without loss of information,''
{\em Quantum Inf Process} {\bf 20}, 65 (2021). 

\bibitem{PCX21}
X. Pan, X. Chen, G. Xu,  H. Ahmad, T. Shang, Z.-P. Li, and Y.-X. Yang,
``Controlled Quantum Network Coding Without Loss of Information,''
{\em CMC-Computers, Materials \& Continua}, {\bf 69}(3), 3967–3979 (2021).

\bibitem{WE21}
N. Walk and J. Eisert,
``Sharing Classical Secrets with Continuous-Variable Entanglement: Composable Security and Network Coding Advantage,''
{\em PRX Quantum}, {\bf 2}, 040339 (2021).

 \bibitem{Childs}
A. M. Childs,  ``Secure assisted quantum computation,''
{\em Quantum Inf. Comput.} {\bf 5} 456 -- 66 (2005).

\bibitem{BFK}
A. Broadbent, J. Fitzsimons and E. Kashefi, 
``Universal blind quantum computation,'' 
{\em Proc. 50th Annual Symp. on Found. of Comput. Sci.} 
pp 517--26 (2009). 

\bibitem{BKBF}
S. Barz, E. Kashefi, A. Broadbent, J. F. Fitzsimons, A. Zeilinger,
and P. Walther, 
``Demonstration of blind quantum computing,'' 
{\em Science} 335 303, 8,  (2012).
\bibitem{MF}
T. Morimae and K. Fujii 
``Blind topological measurement-based quantum computation,'' 
{\em Nat. Commun.} 3 1036 (2012).

\bibitem{Morimae}
T. Morimae, 
``Verification for measurement-only blind quantum computing,'' 
{\em Phys. Rev. A} 89 060302(R) (2014) 

\bibitem{MDF}
A. Mantri, C. A. P. Delgado, and J. F. Fitzsimons, 
``Optimal blind quantum computation,'' 
{\em Phys. Rev. Lett.} 111 230502 (2013).

\bibitem{MF2}
T. Morimae and K. Fujii, 
``Secure entanglement distillation for double-server blind quantum computation,'' 
{\em Phys. Rev. Lett.} 111 020502 (2013).

\bibitem{LCWW}
Q. Li, W. H. Chan, C. Wu, and Z. Wen, 
``Triple-sever blind quantum computation using entanglement swapping,'' 
{\em Phys. Rev. A} 89 040302(R)  (2014).

\bibitem{SZ}
Y.-B. Sheng and L. Zhou, 
``Deterministic entanglement distillation for secure double-server blind quantum
computation,''
{\em Sci. Rep.} 5 7815  (2015).

\bibitem{HM}
M. Hayashi and T. Morimae, 
``Verifiable measurement-only blind quantum computing with stabilizer testing, 
{\em Phys. Rev. Lett.} 115 220502 (2015).

\bibitem{Lo}
H.-K. Lo,
``Insecurity of quantum secure computations,'' 
Physical Review A, 56(2): 1154 -- 1162, Aug 1997. ISSN 1094-1622.


\bibitem{KdW03}
I. Kerenidis and R. de Wolf. ``Exponential lower bound for 2-query locally decodable
codes via a quantum argument,'' {\em Proceedings of 35th ACM STOC}, pp. 106--115, 2003.

\bibitem{KdW04}
I. Kerenidis and R. de Wolf, ``Quantum symmetrically-private information retrieval,''
{\em Information Processing Letters}, vol. 90, pp. 109--114, 2004.


\bibitem{SH19}
S. Song and M. Hayashi, 
``Capacity of Quantum Private Information Retrieval with Multiple Servers,'' 
% {\em Proceedings of 2019 IEEE International Symposium on Information Theory (ISIT)}, pp. 1727--1731, 2019.
                {\em IEEE Transactions on Information Theory}, vol. 67, no. 1, pp. 452--463, 2021.
          

\bibitem{SJ17}
H. Sun and S. Jafar, 
``The capacity of private information retrieval,'' 
{\em IEEE Transactions on Information Theory}, vol. 63, no. 7, pp. 4075--4088, 2017.

\bibitem{SJ17-2}
H. Sun and S. Jafar, ``The Capacity of Symmetric Private Information Retrieval,'' 2016 IEEE Globecom Workshops (GC Wkshps), Washington, DC, 2016, pp. 1--5.

 
\bibitem{SH19-2}
S. Song and M. Hayashi, 
                ``Capacity of Quantum Private Information Retrieval with Collusion of All But One of Servers,''
                {\em IEEE Journal on Selected Areas in Information Theory},  vol. 2, no. 1, pp. 380--390, 2021.
% {\em Proceedings of 2018 IEEE Information Theory Workshop (ITW)}, pp. 1--5, 2019.
                
\bibitem{SH20}
S. Song and M. Hayashi, 
                ``Capacity of Quantum Private Information Retrieval with Colluding Servers,''
                {\em IEEE Transactions on Information Theory}, in press.
% {\em Proceedings of 2020 IEEE International Symposium on Information Theory (ISIT)}, 2020, in press.
                
\bibitem{AHPH20}
M. Allaix, L. Holzbaur, T. Pllaha, and C. Hollanti,
	``Quantum Private Information Retrieval From Coded and Colluding Servers,''
	{\em IEEE Journal on Selected Areas in Information Theory,} vol.~1, no.~2, 2020.
 %{\em arXiv:2001.05883 [cs.IT]}, 2020.
 

\bibitem{ASHPHH21}
 M. Allaix, S. Song, L. Holzbaur, T. Pllaha, M. Hayashi, and C. Hollanti,
``On the Capacity of Quantum Private Information Retrieval from MDS-Coded and Colluding Servers,''
{\em IEEE Journal on Selected Areas in Communications},
vol. 40, no. 3, pp. 885 -- 898, 2022.

 
\bibitem{KL20}
W. Y. Kon and  C. C. W. Lim,
``Provably Secure Symmetric Private Information Retrieval with Quantum Cryptography,''
 {\em Entropy}, vol.~23, no.~1, 54, 2021.

 \bibitem{WKNL21-1}
 C. Wang, W. Y. Kon, H. J. Ng, and C. C. Lim, 
 ``Experimental symmetric private information retrieval with measurement-device-independent quantum network'',
 {\em arXiv preprint arXiv:2109.12827,} 2021.

 \bibitem{WKNL21-2}
 C. Wang, W. Y. Kon, H. J. Ng, and C. C. Lim, ``Experimental symmetric private information retrieval with quantum key distribution,'' {\em Quantum Information and Measurement VI 2021, F. Sciarrino, N. Treps, M. Giustina, and C. Silberhorn, eds., Technical Digest Series}, Optica Publishing Group, 2021.

 
\bibitem{Wie83}
S. Wiesner. Conjugate Coding. SIGACT News, 15(1):78–88, January 1983.

\bibitem{GC01}
D. Gottesman and I. Chuang. Quantum Digital Signatures, 2001, arXiv:
quant-ph/0105032

\bibitem{Moc07}
C. Mochon. Quantum weak coin flipping with arbitrarily small bias, 2007, arXiv:
0711.4114.

\bibitem{CK09}
A. Chailloux and I. Kerenidis. 
``Optimal Quantum Strong Coin Flipping.'' 
In 50th Annual IEEE Symposium on Foundations of Computer Science, FOCS 2009, October
25-27, 2009, Atlanta, Georgia, USA, pages 527 -- 533. IEEE Computer Society, 2009.

\bibitem{ACG+16}
D. Aharonov, A. Chailloux, M. Ganz, I. Kerenidis, and L. Magnin. 
``A Simpler Proof of the Existence of Quantum Weak Coin Flipping with Arbitrarily Small Bias,'' {\em SIAM J. Comput.}, 45(3):
633 -- 679, 2016.

%\bibitem{DNS10}
% F. Dupuis, J.B. Nielsen, L. Salvail, ``Secure two-party quantum evaluation of unitaries against specious adversaries'', in Proceedings of the 30th Annual Conference on Advances in Cryptology, CRYPTO ‘10,
%(Springer, Berlin, 2010), pp. 685--706, 2010.

\bibitem{288}
R. Jozsa and B. Schumacher, ``A new proof of the quantum noiseless coding theorem,'' 
	{\em J. Mod. Opt.,} 41(12), 2343--2349, 1994.


\bibitem{423}
B. Schumacher, ``Quantum coding,'' 
	{\em Phys. Rev. A,} 51, 2738--2747, 1995.
	
\bibitem{290}
R. Jozsa, M. Horodecki, P. Horodecki, and R. Horodecki, ``Universal quantum information compression,'' {\em Phys. Rev. Lett.,} 81, 1714, 1998.

\bibitem{269}
M. Horodecki, ``Limits for compression of quantum information carried by ensembles of mixed states,'' 
	{\em Phys. Rev. A,} 57, 3364--3369, 1998.

\bibitem{35}
H. Barnum, C. M. Caves, C. A. Fuchs, R. Jozsa, and B. Schumacher, ``On quantum coding for ensembles of mixed states,'' 
	{\em J. Phys. A Math. Gen.,} 34, 6767--6785, 2001.

\bibitem{201}
M. Hayashi, ``Exponents of quantum fixed-length pure state source coding,'' {\em Phys. Rev. A,} 66, 032321, 2002.

\bibitem{SJ18}
H. Sun and S. Jafar, 
``The capacity of robust private information retrieval with colluding databases,''
{\em IEEE Transactions on Information Theory}, vol. 64, no. 4, pp. 2361--2370, 2018.


\bibitem{FHGHK17}
R. Freij-Hollanti, O. W. Gnilke, C. Hollanti, and D. A. Karpuk, ``Private information retrieval from coded databases with colluding servers,'' {\em SIAM J. Appl. Algebra Geometry}, vol. 1, no. 1, pp. 647--664, 2017.


%#--
\end{thebibliography}
\end{document}